\documentclass[11pt]{article}

\usepackage[T1]{fontenc}
\usepackage[utf8]{inputenc}
\usepackage{mathpazo}
\usepackage[margin=1.04in,top=0.96in,bottom=1.05in]{geometry}
\usepackage{amsmath,amssymb,amsthm,mathtools,bm}
\usepackage{microtype}
\usepackage{booktabs}
\usepackage{enumitem}
\usepackage{xcolor}
\usepackage{needspace}
\usepackage{mdframed}
\usepackage{tikz}
\usetikzlibrary{arrows.meta,positioning,calc}
\usepackage[square]{natbib}
\usepackage{aliascnt}
\usepackage[colorlinks=true,linkcolor=blue!45!black,citecolor=blue!45!black,urlcolor=blue!45!black]{hyperref}
\usepackage[nameinlink,noabbrev]{cleveref}
\newcommand{\doi}[1]{\href{https://doi.org/#1}{\nolinkurl{doi:#1}}}

\definecolor{Navy}{HTML}{16324F}
\definecolor{Steel}{HTML}{476B82}
\definecolor{Pale}{HTML}{EEF3F6}
\definecolor{Ink}{HTML}{20262C}

\hypersetup{
  pdftitle={Polynomial-Time MIMO Detection at the Maximum-Likelihood Threshold},
  pdfauthor={Dimitris Papailiopoulos},
}
\setlist{itemsep=.25em,topsep=.35em}
\allowdisplaybreaks

\newtheorem{theorem}{Theorem}[section]
\newaliascnt{lemma}{theorem}
\newtheorem{lemma}[lemma]{Lemma}
\aliascntresetthe{lemma}
\newaliascnt{proposition}{theorem}
\newtheorem{proposition}[proposition]{Proposition}
\aliascntresetthe{proposition}
\newaliascnt{corollary}{theorem}
\newtheorem{corollary}[corollary]{Corollary}
\aliascntresetthe{corollary}
\theoremstyle{remark}
\newaliascnt{remark}{theorem}
\newtheorem{remark}[remark]{Remark}
\aliascntresetthe{remark}
\crefname{lemma}{Lemma}{Lemmas}
\Crefname{lemma}{Lemma}{Lemmas}
\crefname{proposition}{Proposition}{Propositions}
\Crefname{proposition}{Proposition}{Propositions}
\crefname{corollary}{Corollary}{Corollaries}
\Crefname{corollary}{Corollary}{Corollaries}
\crefname{remark}{Remark}{Remarks}
\Crefname{remark}{Remark}{Remarks}

\newmdenv[
  backgroundcolor=Pale,
  linecolor=Steel,
  linewidth=.8pt,
  roundcorner=3pt,
  innerleftmargin=10pt,innerrightmargin=10pt,
  innertopmargin=8pt,innerbottommargin=8pt,
  skipabove=8pt,skipbelow=8pt,
  nobreak=true
]{keybox}

\newcommand{\E}{\mathbb E}
\newcommand{\Pp}{\mathbb P}
\newcommand{\one}{\boldsymbol 1}
\newcommand{\zero}{\boldsymbol 0}
\newcommand{\sign}{\operatorname{sign}}
\newcommand{\tr}{\operatorname{tr}}
\newcommand{\dH}{d_{\mathrm H}}

\title{\vspace{-1.2cm}\textcolor{Navy}{\bfseries Polynomial-Time MIMO Detection at the\\ Maximum-Likelihood Threshold}}
\author{Dimitris Papailiopoulos\footnote{The results in this paper were proved by GPT-5.6 and
Claude Fable~5, which also drafted the initial manuscript.  The
author posed the problem, directed several rounds of proof
simplification, verified all mathematical arguments,
edited the manuscript, and takes full responsibility for its
content.}\\Microsoft Research \& University of Wisconsin}
\date{}

\begin{document}
\maketitle
\vspace{-1.2cm}

\begin{abstract}
\noindent We prove that exact block recovery in the square Gaussian binary MIMO
model can be achieved in polynomial time at the same first-order SNR
threshold as exhaustive maximum-likelihood detection.  Specifically, for
\[
\boldsymbol y=\sqrt{\rho/N}\,\boldsymbol H\boldsymbol x^\star+\boldsymbol w,
\qquad \boldsymbol x^\star\in\{\pm1\}^N,
\]
and independent standard Gaussian
\(\boldsymbol H\in\mathbb R^{N\times N}\) and \(\boldsymbol w\),
rounded linear MMSE followed by steepest single-bit descent
recovers \(\boldsymbol x^\star\) with failure probability tending to zero, uniformly
over every transmitted word and every \(\rho\ge2\log N\), using
\(O(N^3)\) unit-cost exact-real arithmetic operations.  The model is a
special case of Gaussian random linear estimation, for which AMP state
evolution and replica/MMSE formulas rigorously characterize fixed-parameter
normalized performance.  Those results predict the same \(2\log N\) scale,
but do not by themselves yield an all-coordinate guarantee in the
dimension-dependent regime considered here.  To the best of our knowledge,
no prior work gives polynomial-time exact block recovery at the ML boundary
for this setting; the closest prior square-system theorem, for the box
relaxation, has first-order threshold \(4\log N\).  The proof places the rounded
LMMSE estimate at sublinear Hamming distance from the truth, and then
establishes, uniformly over every error set the local search can
visit, that some wrong bit offers a quantified cost decrease while an
objective barrier confines the search path.  Conversely, if
\(0<\rho\le2\log N-\log\log N-s_N\) with \(s_N\to\infty\) and
\(s_N=o(\log N)\), then a one-bit neighbor beats the transmitted word with probability
tending to one, so even ML detection fails.  Therefore, the
statistical and polynomial-time exact-recovery thresholds coincide
to first order in the stated arithmetic model.
\end{abstract}

\tableofcontents

\section{Introduction}

\paragraph{The question.}
Binary MIMO detection has a simple formulation: recover a vector of
transmitted signs from noisy linear measurements.  Maximum-likelihood (ML)
detection solves the resulting binary least-squares problem.  The same
objective appears in optimum multiuser detection for synchronous CDMA and is
closely related to closest-vector search on a lattice with a noisy planted
point \citep{Verdu1989,TseHanly1999,FinckePohst1985,AgrellEtAl2002}.  For arbitrary
channel matrices and observations, optimum multiuser detection is NP-hard
\citep{Verdu1989}.  Here, however, the channel is Gaussian.  Worst-case
hardness therefore does not answer the average-case question that motivates
this paper:
\begin{center}
\emph{Can a polynomial-time detector recover the entire transmitted word at
the same SNR at which exhaustive ML detection first succeeds?}
\end{center}

This is deliberately an uncoded exact-detection question.  It does not ask
for the capacity of a coded MIMO channel, where coding changes the operational
problem and permits positive rates at constant SNR
\citep{Telatar1999,TseViswanath2005}.  The \(2\log N\) scale studied below is
the threshold for recovering all \(N\) uncoded binary symbols in the stated
model.

\paragraph{What was known about ML.}

Early average-case studies of ML detection include
\citet{HassibiVikaloAsilomar2001,HassibiVikalo2002,HassibiVikalo2005,HassibiVikalo2005PartII},
who derived the expected complexity of the sphere decoder
(Fincke--Pohst enumeration \citep{FinckePohst1985}), averaged over the
channel and noise, and reported polynomial behavior over broad ranges of practically relevant dimensions
and SNRs.
However, \citet{JaldenOttersten2005} showed that this behavior does not persist
in the growing-dimension limit: at every fixed SNR, the expected
complexity is exponential in the dimension.  A key geometric reason is
that, to contain the transmitted point with nonvanishing probability,
the search sphere must have squared radius proportional to the
dimension.

The first-order SNR scaling for successful ML block recovery was identified
by \citet{HansenEtAl2009}, and subsequent work
\citep[Lemma~IV.2]{HassibiEtAl2014} states the condition explicitly as
\[
\rho> 2\log N+f(N),
\qquad f(N)\longrightarrow\infty,
\]
under which the ML block-error probability vanishes.  At
the threshold scale \(\rho=\Theta(\log N)\), the finite-dimensional
lower bound underlying Eq.~(29) of \citet{JaldenOttersten2005}
specializes in the present normalization to
\[
 C(N)\ge 2^{N/(4\rho+2)}-1
      =\exp\{\Theta(N/\log N)\},
\]
which is superpolynomial but subexponential in \(N\).  Under only
\(\rho=O(\log N)\), this lower bound is at least
\(\exp\{\Omega(N/\log N)\}\) and can be exponential when \(\rho=O(1)\).
These works established an ML achievability result but did not prove a
matching converse.  A direct calculation with the \(N\) one-bit neighbors
predicts the same first-order constant.  Let
\(Q(t)=\Pr\{Z\ge t\}\), where \(Z\sim\mathcal N(0,1)\).  A given neighbor
then beats the transmitted word with probability
asymptotic to \(Q(\sqrt{\rho})\) in the relevant regime.  The expected
number of such neighbors is therefore approximately
\(NQ(\sqrt{\rho})\), and the transition
\(NQ(\sqrt{\rho})\asymp1\) occurs at
\[
\rho=2\log N-\log\log N+O(1).
\]
Our constructive theorem succeeds uniformly for every
\(\rho\ge2\log N\), including the boundary itself.  First-order ML
achievability was already known from the works just cited; the new point is
that an explicit polynomial-time detector reaches the same scale.
\paragraph{Polynomial-time detectors and relaxations.}
Among convex relaxations, semidefinite and box relaxations provide the
closest prior comparisons.  For semidefinite relaxation,
\citet{KisialiouLuo2010,So2010} establish probabilistic approximation
guarantees for the objective value, \citet{JaldenOttersten2008} show
that the relaxation attains the maximum-likelihood diversity order in
the fixed-dimensional high-SNR limit, and
\citet{LuLiuZhangZhang2019,JiangLiu2021} give explicit deterministic
tightness conditions.  None of these yields exact block recovery at
any logarithmic SNR scale in the square real Gaussian model.  For the
box relaxation, \citet{ThrampoulidisXuHassibi2018} derive precise
large-system symbol-error asymptotics at fixed SNR, and, extending the
analysis to the vanishing-noise regime, \citet{HuLu2020} prove a
Poisson limit for the number of errors and a sharp block-recovery
transition: in the normalization used here, their square-system
first-order threshold is \(4\log N\), with refined window
\(4\log N-2\log\log N+O(1)\), a factor of two above maximum likelihood
at first order.
To our knowledge, this is the only previous theorem establishing
polynomial-time exact block recovery at logarithmic SNR in this
model.

The algorithm analyzed here is a steepest-selection variant of
likelihood-ascent search (LAS).  One-symbol LAS iteratively accepts
likelihood-improving bit changes and has long been studied as a
low-complexity approximation to ML detection
\citep{Sun1998,VardhanEtAl2008}.  Most directly,
\citet{MohammedChockalingamRajan2008} analyze one-symbol LAS in the square
large-system limit, allowing MMSE initialization and using it in their
simulations, and state asymptotic agreement with ML error performance.
Linear-detector initialization, including MMSE, is also studied empirically
for sequential LAS by \citet{SilvaMarinelloAbrao2018}.
More precisely, Theorem~1 of \citet{MohammedChockalingamRajan2008} is
pointwise for a fixed candidate \(d\).  Their Theorem~2 then substitutes the
data-dependent LAS output \(d_{\mathrm{LAS}}\) without a simultaneous bound
over the possible outputs.  Such simultaneous control is needed for that
adaptive step, and it is supplied here by the uniform landscape event.
Their fixed-parameter analysis also does not establish the growing-SNR
block-recovery statement, confinement, or a high-probability iteration
bound proved here.

Another direct square-system comparison is the projected-Jacobi detector of
\citet{LiuMaTafazolli2023}, whose title-level claim is maximum-likelihood
performance with square-order complexity.  Its analysis concerns fixed-SNR
pairwise and symbol error, explicitly sets aside the iteration dynamics, and
uses a one-error input/output simplification in its key derivation.  It does
not give a total-iteration bound or a worst-word exact-block guarantee in the
triangular regime \(\rho_N\asymp\log N\).

Several rigorous algorithmic results are close in a different direction.
The generalized power method of \citet{LiuYueSoMa2017} has a finite-iteration
exact-recovery guarantee, but its random-channel theorem requires
\(m/n\ge20/\beta^2\) with \(\beta<1/4\), hence a strongly overdetermined
rather than square system.  There are also exact polynomial-time ML results
under structural restrictions that are absent here.
\citet{PapailiopoulosKarystinos2010} give a polynomial-time exact ML detector
for noncoherent OSTBC sequence detection when the relevant channel-covariance
rank is fixed as the sequence length grows.  The zonotope algorithm of
\citet{KhanSlock2003} computes exact ML in polynomial time when the receive
rank is fixed.  Neither result remains polynomial when the rank controlling
the enumeration grows with the square dimension \(N\).

\paragraph{Relation to AMP and statistical physics.}
Message-passing and statistical-physics analyses give a rigorous and
important account of a different asymptotic question.  After dividing the
observation by \(\sqrt{\rho}\), and placing an i.i.d. Rademacher prior on the
coordinates of \(\boldsymbol x^\star\), the present model is exactly the
scalar \(B=1\), square \(\alpha=1\) Gaussian random-linear-estimation model
with noise variance \(\Delta=1/\rho\).  AMP was introduced for Gaussian
compressed sensing by \citet{DonohoMalekiMontanari2009}; the state-evolution
theorems of \citet{BayatiMontanari2011,JavanmardMontanari2013} rigorously
track empirical observables of its iterates when the iteration index and
model parameters are held fixed as the dimensions tend to infinity.
Replica calculations predicted the corresponding Bayes-optimal formulas for
CDMA and Gaussian linear estimation \citep{Tanaka2002,GuoVerdu2005}.
For scalar discrete priors, the replica-symmetric mutual-information and
MMSE formulas, together with AMP optimality outside the hard phase, were
subsequently established rigorously under the at-most-three-stationary-points
condition on the replica-symmetric potential stated by
\citet{BarbierMacrisDiaKrzakala2020}; the Rademacher model here falls under
that conclusion when the stated condition holds.  Related
results for Gaussian linear and generalized-linear models appear in
\citet{ReevesPfister2019,BarbierEtAl2019}.

For the low-noise fixed point relevant here, state evolution describes the
empirical law of a typical coordinate through an effective scalar Gaussian
channel and predicts a hard-decision error of order \(Q(\sqrt{\rho})\).  The
heuristic block calculation
\(NQ(\sqrt{\rho})\asymp1\) therefore already predicts both the leading
\(2\log N\) scale and its \(-\log\log N\) correction.  This is genuine prior
understanding, not merely a formal analogy: the scalar-channel
state-evolution picture predicts why the \(2\log N\) scale should appear.
For a recent discussion of this prediction and its relationship to the
present block-recovery question, see \citet{Krzakala2026MIMOAMP}.  What those results do
not automatically provide is the theorem proved here.  Fixed-parameter
normalized error rates cannot distinguish zero errors from one error or from
\(o(N)\) errors.  One direct AMP route based on a coordinatewise union bound
would need triangular-array control for \(\rho_N\asymp\log N\), simultaneous
\(o(1/N)\)-scale coordinate tails, and control of any iteration count that
grows with \(N\); a different route would need comparably strong
all-coordinate control.  MIMO
specializations such as LAMA rigorously characterize large-system symbol-error
rates and identify fixed-parameter regimes of individually optimal detection
under additional fixed-point conditions \citep{JeonEtAl2018}; they do not
state this all-coordinate growing-SNR guarantee.

The CLuP line of \citet{Stojnic2019,Stojnic2019Complexity} analyzes
fixed-SNR bit-error performance and individual iterations through
random-duality calculations and reports a small empirical iteration count.
It does not establish exact block recovery at scaling SNR or a polynomial
bound on the total number of iterations required for the present guarantee.
Recent leave-one-out theory gives nonasymptotic coordinatewise
representations for AMP \citep{BaoHanXu2025} and may support a different
proof route, but it does not state the present block theorem.  Such a
specialization would still have to handle the triangular growing-SNR
regime, a possibly growing iteration count, leave-column variance
calibration, and simultaneous Gaussian tails at the \(1/N\) scale.  We
are not aware of a published specialization carrying out these steps at
\(\rho=2\log N\).  MCMC analyses show that, at a suitably
chosen temperature, the optimum receives inverse-polynomial (rather
than exponentially small) stationary mass on the $2\log N$ SNR scale, but
polynomial mixing for square Gaussian instances remains open
\citep{HassibiDimakisPapailiopoulos2010,HassibiEtAl2014}.

\paragraph{Our contributions.}
Our main result establishes that there is no first-order
computational--statistical gap for exact recovery in the square
Gaussian model.  Rounded linear MMSE followed by a steepest-selection
variant of one-bit likelihood-ascent search
\citep{Sun1998,VardhanEtAl2008} recovers the
transmitted word throughout the regime $\rho\ge2\log N$, including
the boundary $\rho=2\log N$ itself.  More precisely, writing
$\widehat{\boldsymbol{x}}_N$ for the detector output,
\[
\lim_{N\to\infty}
\sup_{\substack{
\boldsymbol{x}^{\star}\in\{\pm1\}^N\\
\rho\ge 2\log N}}
\Pr_{\boldsymbol H,\boldsymbol w}\!\left\{
\widehat{\boldsymbol{x}}_N
\neq
\boldsymbol{x}^{\star}
\right\}
=0.
\]
The detector accepts at most \(O(N\log N)\) bit flips and uses \(O(N^3)\)
unit-cost exact-real arithmetic operations, with the dense LMMSE solve
dominating the complexity.

The detector itself is classical.  Closely related LMMSE-initialized
one-symbol search procedures have been studied through asymptotic,
symbol-error, and local-optimality analyses
\citep{MohammedChockalingamRajan2008,Sun2009}.  Its two-stage
architecture---first obtain a partial estimate, then refine it through
local corrections---is also familiar from exact community recovery and
low-rank matrix completion
\citep{AbbeBandeiraHall2016,KeshavanMontanariOh2010}.  The contribution
here is the threshold-scale uniform analysis of that architecture for the
square Gaussian MIMO objective.
Fixed-candidate and empirical symbol-error arguments do not control the
data-dependent sequence of points visited by a local-search algorithm.
We instead prove geometric bounds that hold simultaneously for every
error set in a sufficiently large Hamming ball.  These uniform bounds
show that every nontruth point in the ball admits a decreasing
bit flip and that the decreasing-cost trajectory cannot leave
the ball.  They therefore apply directly to the local search path,
without conditioning on that path, and yield an explicit iteration
bound at \(\rho=2\log N\).

We complement the constructive result with an ML converse.  Let
\(s_N\to\infty\) with \(s_N=o(\log N)\).  Uniformly over deterministic
transmitted words and
\[
0<\rho\le
2\log N-\log\log N-s_N,
\]
a one-bit neighbor has strictly smaller least-squares cost than the
transmitted word with probability tending to one.  Consequently, even
ML detection fails in this regime.  The achievability
and converse boundaries differ by only \(o(\log N)\), establishing the
sharp first-order threshold \(2\log N\).  A standard uniform-prior MAP
argument, deferred to Appendix~\ref{app:decision}, upgrades this ML converse
to a minimax lower bound for arbitrary detectors.

\paragraph{Organization.}
The proof is presented in dependency order.  We first establish that
rounded LMMSE produces an estimate at sublinear Hamming distance from
the transmitted word.  We then derive the exact identities governing
a one-bit flip and prove two uniform geometric estimates: a coarse
estimate for the outer portion of the Hamming ball and a sharper
estimate near the truth.  Once these probabilistic events have been
established, the remainder of the recovery argument is deterministic:
the objective decreases, the trajectory cannot cross the boundary of
the controlled ball, and every nontruth point offers an admissible
decreasing move.  We prove the ML converse after completing the
constructive argument.  Each lemma is introduced by the question it
answers and by its role in the subsequent step.
Appendix~\ref{app:adjacent} collects exact non-MIMO formulations of the
model and comparisons with the closest results in adjacent literatures.

\section{System Setting and Main Result}\label{sec:model-result}

Let
\begin{equation}\label{eq:model}
 \boldsymbol y
 =
 \sqrt{\frac{\rho}{N}}\,
 \boldsymbol H\boldsymbol x^\star+\boldsymbol w,
 \qquad
 \boldsymbol x^\star\in\{\pm1\}^N,
\end{equation}
where \(\boldsymbol H\in\mathbb R^{N\times N}\) has independent
\(\mathcal N(0,1)\) entries and
\(\boldsymbol w\sim\mathcal N(\zero,\boldsymbol I_N)\) is independent of
\(\boldsymbol H\).  The transmitted vector
\(\boldsymbol x^\star\) is deterministic and arbitrary.  The receiver
knows \((\boldsymbol y,\boldsymbol H,\rho)\).  With this normalization,
\(\rho\) is the average received SNR per observation.

Throughout, \(N\in\{2,3,\ldots\}\) and \(\rho>0\); every asymptotic
statement concerns \(N\to\infty\).  All logarithms are natural, and all probabilities refer only to the
randomness of \((\boldsymbol H,\boldsymbol w)\), unless explicitly stated
otherwise.  Complexity is measured
in a unit-cost exact-real arithmetic model: additions, subtractions,
multiplications, divisions, comparisons, and evaluation of the known
scale \(\sqrt{\rho/N}\) are exact unit-cost operations.  We do not claim
a Turing bit-complexity or finite-precision bound.

The least-squares objective is
\begin{equation}\label{eq:objective}
 f(\boldsymbol x)
 =
 \left\|
 \boldsymbol y-
 \sqrt{\frac{\rho}{N}}\,
 \boldsymbol H\boldsymbol x
 \right\|_2^2,
 \qquad
 \boldsymbol x\in\{\pm1\}^N.
\end{equation}
An ML detector returns any minimizer
\begin{equation}\label{eq:ml-detector}
 \widehat{\boldsymbol x}_{\mathrm{ML}}
 \in
 \arg\min_{\boldsymbol x\in\{\pm1\}^N}
 f(\boldsymbol x).
\end{equation}

The polynomial-time detector analyzed in this paper has two stages.

\begin{keybox}
\textbf{Stage 1: rounded LMMSE.}
Compute
\begin{equation}\label{eq:detector-lmmse}
 \boldsymbol b
 =
 \sqrt{\frac{N}{\rho}}
 \left(
 \boldsymbol H^{\mathsf T}\boldsymbol H
 +\frac{N}{\rho}\boldsymbol I_N
 \right)^{-1}
 \boldsymbol H^{\mathsf T}\boldsymbol y,
 \qquad
 \boldsymbol x^{\mathrm{init}}
 =
 \sign(\boldsymbol b).
\end{equation}
A zero coordinate is rounded to \(+1\).  Such a tie occurs with
probability zero under the Gaussian model, so the convention does not
affect the theorem.  The positive prefactor \(\sqrt{N/\rho}\) does not
affect the rounded vector, but it is retained so that
\(\boldsymbol b\) is the correctly normalized unit-covariance LMMSE
estimate under an i.i.d. Rademacher prior.  In the deterministic-word
formulation of the theorem, it is equivalently an LMMSE-form ridge estimate.

\medskip
\textbf{Stage 2: steepest single-bit descent with a deterministic flip limit.}
For the current vector \(\boldsymbol x\), let
\(\boldsymbol x^{(i)}\) denote the vector obtained by flipping
coordinate \(i\), and define its objective decrease by
\[
 \Delta_i(\boldsymbol x)
 =
 f(\boldsymbol x)-f(\boldsymbol x^{(i)}).
\]
If
\[
 \max_{1\le i\le N}\Delta_i(\boldsymbol x)<\frac1N,
\]
return the current vector.  Otherwise, flip the smallest index attaining
the largest gain.  Define the deterministic flip limit
\[
 L_N=\min\{\ell\in\mathbb N:\ell\ge1,\ 2^\ell\ge N\},
 \qquad M_N=4NL_N.
\]
Thus \(L_N=\max\{1,\lceil\log_2N\rceil\}\), and it can be computed by
repeated integer doubling and comparison within the stated arithmetic
model.
After accepting the \(M_N\)-th flip, return the resulting vector
unconditionally.  Denote the vector returned in either case by
\(\widehat{\boldsymbol x}\).
\end{keybox}

Once
\(\boldsymbol H^{\mathsf T}\boldsymbol H\) and the initial gains have
been computed, all \(N\) gains can be updated in \(O(N)\) operations
after each accepted flip.  Since the flip limit is \(O(N\log N)\), the descent
stage costs \(O(N^2\log N)\) operations.  Forming and solving the dense
LMMSE system costs \(O(N^3)\) operations and therefore still dominates
the total complexity.

\Needspace{8\baselineskip}
\begin{theorem}[Sharp first-order exact-recovery threshold]
\label{thm:main}
The following two statements hold.
\begin{enumerate}[label=\textup{(\alph*)}]

\item The detector above obeys the uniform error-probability guarantee
\begin{equation}\label{eq:uniform-risk}
 \lim_{N\to\infty}
 \sup_{\substack{
       \boldsymbol x^\star\in\{\pm1\}^N\\
       \rho\ge2\log N}}
 \Pp_{\boldsymbol H,\boldsymbol w}
 \!\left\{
 \widehat{\boldsymbol x}\ne\boldsymbol x^\star
 \right\}
 =0.
\end{equation}
Its worst-case cost is \(O(N^3)\) unit-cost exact-real arithmetic
operations.

\item Let \(s_N\to\infty\) with \(s_N=o(\log N)\).  Then
\begin{equation}\label{eq:uniform-converse-risk}
 \lim_{N\to\infty}
 \sup_{\substack{
       \boldsymbol x^\star\in\{\pm1\}^N\\
       0<\rho\le
       2\log N-\log\log N-s_N}}
 \Pp_{\boldsymbol H,\boldsymbol w}
 \!\left\{
 \boldsymbol x^\star
 \text{ is an ML minimizer of }f
 \right\}
 =0.
\end{equation}
Consequently, the block-error probability of every ML detector tends
to one uniformly throughout this lower regime.

\end{enumerate}
\end{theorem}

For all sufficiently large \(N\), the upper endpoint in
\eqref{eq:uniform-converse-risk} is positive, so its parameter set is
nonempty.

The supremum in \eqref{eq:uniform-risk} concerns the worst error
probability over the displayed deterministic parameters.  It does not
assert that one realization of
\((\boldsymbol H,\boldsymbol w)\) succeeds simultaneously for every
transmitted word and every SNR.  Likewise,
\eqref{eq:uniform-converse-risk} is a uniform probability statement,
not a simultaneous assertion over all words and SNR values on one
realization.

The gap between the two proven boundaries is
\(\log\log N+s_N=o(\log N)\).  The theorem therefore identifies the
common first-order constant \(2\), but does not claim to determine the
complete lower-order transition window.  The decision-theoretic minimax
interpretation is recorded separately in Appendix~\ref{app:decision}.

\section{Roadmap of the achievability proof}

The proof is easiest to read as a sequence of questions.  How close does
rounded LMMSE get to the truth?  What does the objective landscape look like
inside that neighborhood?  Once the landscape bounds hold, why can the greedy
path neither escape, stop incorrectly, nor exhaust its allowed number of
flips before recovery?  The three panels below answer these questions in the
same order as the formal proof.

\subsubsection*{A. Rounded LMMSE comes within sublinear distance of the ground truth}

\begin{center}
\begin{tikzpicture}[
  roadstep/.style={draw=Steel,rounded corners=2pt,fill=Pale,
    text width=12.6cm,align=center,inner sep=7pt,font=\small},
  arr/.style={-{Stealth[length=2.1mm]},line width=.9pt,draw=Navy}
]
\node[roadstep] (mse) {
  \textbf{1. Convert estimation error into one trace.} With
  \(\boldsymbol B_\rho=(\boldsymbol I_N+\frac\rho N
  \boldsymbol H^{\mathsf T}\boldsymbol H)^{-1}\), the exact LMMSE identity is
  \[
    \E_{\boldsymbol H,\boldsymbol w}
    \|\boldsymbol b-\boldsymbol x^\star\|_2^2
    =\E_{\boldsymbol H}\tr(\boldsymbol B_\rho).
  \]
  {\footnotesize \Cref{lem:mse}.}
};
\node[roadstep,below=5mm of mse] (trace) {
  \textbf{2. Bound that trace.} A leave-one-row identity and a direct
  expectation bound give, uniformly for \(\rho\ge2\log N\),
  \[
    \E\tr(\boldsymbol B_\rho)
    \le\frac{2N}{\sqrt{\log N}}.
  \]
  {\footnotesize \Cref{lem:trace-identity,lem:trace-bound}.}
};
\node[roadstep,below=5mm of trace] (warm) {
 \textbf{3. Convert squared error into sign errors.}
Define
\[
\boldsymbol x^{\mathrm{init}}
=\sign(\boldsymbol b),
\qquad
K=
d_{\mathrm H}
(\boldsymbol x^{\mathrm{init}},\boldsymbol x^\star).
\]
If coordinate \(j\) is rounded incorrectly, then
$
|b_j-x_j^\star|\ge 1.
$
Therefore every wrong coordinate contributes at least one to the squared
error, and hence
$
K
\le
\|\boldsymbol b-\boldsymbol x^\star\|_2^2.
$
Consequently, Markov's inequality and the preceding trace bound give
\[
\begin{aligned}
\Pr\!\left\{
K>\left\lceil\frac{N}{(\log N)^{1/4}}\right\rceil
\right\}
&\le
\Pr\!\left\{
\|\boldsymbol b-\boldsymbol x^\star\|_2^2
>
\frac{N}{(\log N)^{1/4}}
\right\}\\
&\le
\frac{
\E_{\boldsymbol H,\boldsymbol w}
\|\boldsymbol b-\boldsymbol x^\star\|_2^2
}{
N/(\log N)^{1/4}
}
\le
\frac{2}{(\log N)^{1/4}}
\longrightarrow 0.
\end{aligned}
\]
  {\footnotesize \Cref{prop:warm}.}
};
\draw[arr] (mse)--(trace);
\draw[arr] (trace)--(warm);
\end{tikzpicture}
\end{center}

\clearpage
\subsubsection*{B. The landscape near the ground truth has favorable geometry}

{\footnotesize
For a current vector \(\boldsymbol x\), set
\(S=\{i:x_i\ne x_i^\star\}\), \(k=|S|\), and
\(\boldsymbol h_i=x_i^\star\boldsymbol H_{:,i}\), using column-sign
symmetry to normalize \(\boldsymbol x^\star=\one\).  This sign is needed
in the noise terms.  The integer endpoints below match the formal proof.
\par}
\vspace{-2mm}
\begin{center}
\begin{tikzpicture}[
  top/.style={draw=Steel,rounded corners=2pt,fill=Pale,
    text width=14.75cm,align=center,inner sep=5pt,font=\footnotesize,
    execute at begin node={\setlength{\abovedisplayskip}{2pt}%
      \setlength{\belowdisplayskip}{2pt}%
      \setlength{\jot}{1pt}}},
  region/.style={draw=Steel,rounded corners=2pt,fill=Pale,
    text width=6.82cm,align=left,inner sep=5pt,font=\footnotesize,
    execute at begin node={\setlength{\abovedisplayskip}{2pt}%
      \setlength{\belowdisplayskip}{2pt}%
      \setlength{\jot}{1pt}}},
  merge/.style={draw=Navy,rounded corners=6pt,fill=Navy!6,
    text width=14.75cm,align=left,inner sep=5pt,font=\footnotesize,
    execute at begin node={\setlength{\abovedisplayskip}{2pt}%
      \setlength{\belowdisplayskip}{2pt}%
      \setlength{\jot}{1pt}}},
  event/.style={draw=Navy,rounded corners=6pt,fill=Navy!9,
    text width=14.75cm,align=center,inner sep=4pt,font=\footnotesize,
    execute at begin node={\setlength{\abovedisplayskip}{1.5pt}%
      \setlength{\belowdisplayskip}{1.5pt}}},
  arr/.style={-{Stealth[length=1.7mm]},line width=.75pt,draw=Navy}
]

\node[top] (algebra) {
  \textbf{Exact flip algebra: the two quantities we need}\\[-1pt]
  \[
  \begin{aligned}
  f(\boldsymbol x)-f(\boldsymbol x^\star)
  &=\frac{4\rho}{N}\left\|\sum_{i\in S}\boldsymbol h_i\right\|_2^2
    +4\sqrt{\frac{\rho}{N}}\,
      \boldsymbol w^{\mathsf T}\sum_{i\in S}\boldsymbol h_i,\\[-1pt]
  \sum_{i\in S}\Delta_i(\boldsymbol x)
  &=4\!\left[\frac{\rho}{N}\!\left(
      2\left\|\sum_{i\in S}\boldsymbol h_i\right\|_2^2
      -\sum_{i\in S}\|\boldsymbol h_i\|_2^2\right)
      +\sqrt{\frac{\rho}{N}}\,
       \boldsymbol w^{\mathsf T}\sum_{i\in S}\boldsymbol h_i\right].
  \end{aligned}
  \]
  The first line is the excess cost.  The second sums the gains from
  correcting all \(k\) wrong bits: if it is at least \(ak\), some wrong
  bit has gain at least \(a\), so the maximal available gain is no smaller;
  if \(a\ge1/N\), the detector accepts such a gain.
  {\footnotesize
  \Cref{lem:objective-gap,lem:total-wrong-gain,lem:average-wrong-gain}.}
};

\node[region,below=3mm of algebra,xshift=-3.65cm] (outer) {
  \textbf{Outer region: coarse bounds suffice}
  \[
    \left\lceil\frac{N}{(\log N)^3}\right\rceil<k
    \le\left\lceil\frac{4N}{(\log N)^{1/4}}\right\rceil.
  \]
  With probability tending to one, simultaneously for every such set,
  \[
  \begin{aligned}
  \frac52\rho k&\le f(\boldsymbol x)-f(\boldsymbol x^\star)
    \le\frac{11}{2}\rho k,\\[-1pt]
  \max_{i\in S}\Delta_i(\boldsymbol x)&\ge\frac\rho2.
  \end{aligned}
  \]
  \textbf{Uniformity mechanism.}  Conditional on the uniform channel
  bounds, fixed-set failure is at most
  \[
    2e^{-k\log N/80}.
  \]
  There are exactly \(\binom Nk\) sets and
  \[
    \log\binom Nk\le k\log\frac{eN}{k}
      \le k(1+3\log\log N).
  \]
  This exponent is \(o(k\log N)\), so the union bound tends to zero.
  {\footnotesize
  \Cref{lem:outer-channel,lem:outer-noise,lem:outer-gain}.}
};

\node[region,below=3mm of algebra,xshift=3.65cm] (inner) {
  \textbf{Inner region: sharp tails are required}
  \[
    1\le k\le\left\lceil\frac{N}{(\log N)^3}\right\rceil.
  \]
  With probability tending to one, simultaneously for every such set,
  \[
  \begin{aligned}
  \frac{12\rho k}{(\log N)^{5/4}}
    &\le f(\boldsymbol x)-f(\boldsymbol x^\star)\le8\rho k,\\[-1pt]
    \max_{i\in S}\Delta_i(\boldsymbol x)
    &\ge\frac{4\rho}{(\log N)^{5/4}}.
  \end{aligned}
  \]
  \textbf{Uniformity mechanism.}  Conditional on the sharper channel
  bounds, fixed-set failure is at most
  \[
    \frac{2}{\sqrt{k\log N}}\,
    N^{-k(1-9/(\log N)^{5/4})}.
  \]
  Since \(\binom Nk\le N^k/k!\), the leading powers cancel; the residual
  correction is at most \(2^k\), and the full union bound is at most
  \[
  \frac{2}{\sqrt{\log N}}
  \sum_{k=1}^{\lceil N/(\log N)^3\rceil}
    \frac{2^k}{k!\sqrt{k}}
  \le\frac{2(e^2-1)}{\sqrt{\log N}}\longrightarrow0.
  \]
  {\footnotesize
  \Cref{lem:inner-channel,lem:inner-noise,lem:inner-gain}.}
};

\draw[arr] (algebra.south) -- (outer.north);
\draw[arr] (algebra.south) -- (inner.north);

\node[merge,below=3mm of inner.south,xshift=-3.65cm] (combined) {
  \centering\textbf{What the two regions give together}\par
  \raggedright
  The ranges meet at \(\lceil N/(\log N)^3\rceil\), so every nontruth vector with
  \(1\le d_{\mathrm H}(\boldsymbol x,\boldsymbol x^\star)
  \le\lceil4N/(\log N)^{1/4}\rceil\) is covered uniformly before the algorithm runs.
  Consequently,
  \[
  \begin{array}{@{}ll@{}}
  \text{decreasing move:}&
  \exists i\in S:\quad
  \Delta_i(\boldsymbol x)\ge\dfrac{4\rho}{(\log N)^{5/4}},\\[3pt]
  \text{outer barrier:}&
  d_{\mathrm H}(\boldsymbol x,\boldsymbol x^\star)
  =\left\lceil\dfrac{4N}{(\log N)^{1/4}}\right\rceil
  \Rightarrow
  f(\boldsymbol x)-f(\boldsymbol x^\star)
  \ge\dfrac52\rho
  \left\lceil\dfrac{4N}{(\log N)^{1/4}}\right\rceil.
  \end{array}
  \]
  Thus there is no wrong stopping point in the ball; a decreasing path that
  starts at most at
  \(8\rho\lceil N/(\log N)^{1/4}\rceil\) cannot cross its boundary.
  The statements are uniform, so they apply along the adaptive path.
  {\footnotesize
  \Cref{lem:outer-gain,lem:inner-gain}; \Cref{prop:descent}.}
};

\draw[arr] (outer.south) |- ([xshift=-3.5cm]combined.north);
\draw[arr] (inner.south) |- ([xshift=3.5cm]combined.north);

\node[event,below=2.5mm of combined] (joint) {
  \textbf{Warm start \(\cap\) landscape.}\;
  \(\Pr\{\text{both hold}\}\ge
  1-\Pr\{\text{warm start fails}\}
  -\Pr\{\text{landscape fails}\}\longrightarrow1.\)
  No independence is required.
};

\draw[arr] (combined)--(joint);

\end{tikzpicture}
\end{center}
\clearpage
\subsubsection*{C. A decreasing objective forces exact recovery}

Fix a realization on which the warm-start and landscape events both hold.
The remainder of the argument is deterministic.

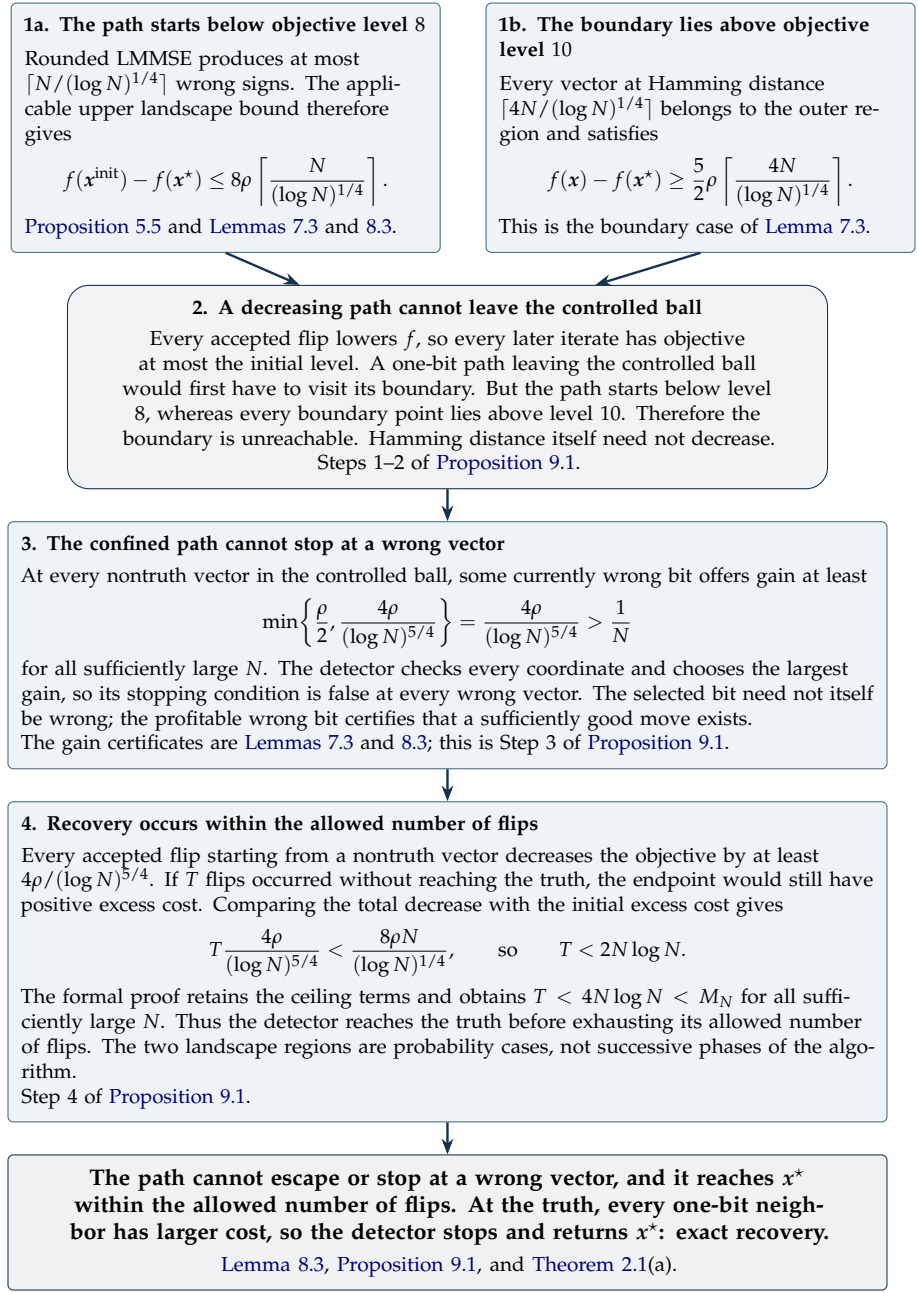
\begin{center}
\begin{tikzpicture}[
  scale=.85,transform shape,
  fact/.style={draw=Steel,rounded corners=2pt,fill=Pale,
    text width=6.25cm,align=left,inner sep=6pt,font=\footnotesize,
    execute at begin node={\setlength{\abovedisplayskip}{4pt}%
      \setlength{\belowdisplayskip}{4pt}}},
  roadwide/.style={draw=Steel,rounded corners=2pt,fill=Pale,
    text width=13.25cm,align=left,inner sep=6pt,font=\footnotesize,
    execute at begin node={\setlength{\abovedisplayskip}{4pt}%
      \setlength{\belowdisplayskip}{4pt}}},
  merge/.style={draw=Navy,rounded corners=8pt,fill=Navy!6,
    text width=11.4cm,align=center,inner sep=6pt,font=\footnotesize,
    execute at begin node={\setlength{\abovedisplayskip}{4pt}%
      \setlength{\belowdisplayskip}{4pt}}},
  result/.style={draw=Navy,rounded corners=2pt,fill=Navy!7,
    text width=13.25cm,align=center,inner sep=6pt,font=\small},
  arr/.style={-{Stealth[length=2.1mm]},line width=.9pt,draw=Navy}
]

\node[fact] (startcost) {
  \textbf{1a. The path starts below objective level \(8\)}\\[4pt]
  Rounded LMMSE produces at most
  \(\lceil N/(\log N)^{1/4}\rceil\) wrong signs.  The
  applicable upper landscape bound therefore gives
  \[
    f(\boldsymbol x^{\mathrm{init}})-f(\boldsymbol x^\star)
    \le8\rho\left\lceil\frac{N}{(\log N)^{1/4}}\right\rceil.
  \]
  {\footnotesize
   \Cref{prop:warm,lem:inner-gain,lem:outer-gain}.}
};

\node[fact,right=7mm of startcost] (wall) {
  \textbf{1b. The boundary lies above objective level \(10\)}\\[4pt]
  Every vector at Hamming distance
  \(\lceil4N/(\log N)^{1/4}\rceil\) belongs to the
  outer region and satisfies
  \[
    f(\boldsymbol x)-f(\boldsymbol x^\star)
    \ge\frac52\rho
    \left\lceil\frac{4N}{(\log N)^{1/4}}\right\rceil.
  \]
  {\footnotesize This is the boundary case of
   \Cref{lem:outer-gain}.}
};

\node[merge,below=5mm of startcost,xshift=3.45cm] (confine) {
  \textbf{2. A decreasing path cannot leave the controlled ball}\\[3pt]
  Every accepted flip lowers \(f\), so every later iterate has objective at
  most the initial level.  A one-bit path leaving the controlled ball would
  first have to visit its boundary.  But the path starts below level \(8\),
  whereas every boundary point lies above level \(10\).  Therefore the
  boundary is unreachable.  Hamming distance itself need not decrease.\par
  {\footnotesize Steps 1--2 of \Cref{prop:descent}.}
};

\draw[arr] (startcost.south) -- ([xshift=-23mm]confine.north);
\draw[arr] (wall.south) -- ([xshift=23mm]confine.north);

\node[roadwide,below=5mm of confine] (nostop) {
  \textbf{3. The confined path cannot stop at a wrong vector}\\[3pt]
  At every nontruth vector in the controlled ball, some currently wrong bit
  offers gain at least
  \[
    \min\!\left\{\frac\rho2,
      \frac{4\rho}{(\log N)^{5/4}}\right\}
    =\frac{4\rho}{(\log N)^{5/4}}>\frac1N
  \]
  for all sufficiently large \(N\).  The detector checks every coordinate
  and chooses the largest gain, so its stopping condition is false at every
  wrong vector.  The selected bit need not itself be wrong; the profitable
  wrong bit certifies that a sufficiently good move exists.\par
  {\footnotesize The gain certificates are
   \Cref{lem:inner-gain,lem:outer-gain}; this is Step 3 of
   \Cref{prop:descent}.}
};

\node[roadwide,below=5mm of nostop] (count) {
  \textbf{4. Recovery occurs within the allowed number of flips}\\[3pt]
  Every accepted flip starting from a nontruth vector decreases the objective
  by at least \(4\rho/(\log N)^{5/4}\).  If \(T\) flips occurred without
  reaching the truth, the endpoint would still have positive excess cost.
  Comparing the total decrease with the initial excess cost gives
  \[
    T\frac{4\rho}{(\log N)^{5/4}}
    <\frac{8\rho N}{(\log N)^{1/4}},
    \qquad\text{so}\qquad T<2N\log N.
  \]
  The formal proof retains the ceiling terms and obtains
  \(T<4N\log N<M_N\) for all sufficiently large \(N\).  Thus the detector
  reaches the truth before exhausting its allowed number of flips.  The two
  landscape regions are probability cases, not successive phases of the
  algorithm.\par
  {\footnotesize Step 4 of \Cref{prop:descent}.}
};

\node[result,below=5mm of count] (recovery) {
  \textbf{The path cannot escape or stop at a wrong vector, and it reaches
  \(\boldsymbol x^\star\) within the allowed number of flips.  At the truth,
  every one-bit neighbor has larger cost, so the detector stops and returns
  \(\boldsymbol x^\star\): exact recovery.}\\[2pt]
  {\footnotesize \Cref{lem:inner-gain}, \Cref{prop:descent}, and
   \Cref{thm:main}\textup{(a)}.}
};

\draw[arr] (confine)--(nostop);
\draw[arr] (nostop)--(count);
\draw[arr] (count)--(recovery);

\end{tikzpicture}
\end{center}

\section{Standard inequalities used in the proof}

Throughout the paper, an event holds ``with probability tending to one'' if
its probability, under the joint randomness of \(\boldsymbol H\) and
\(\boldsymbol w\), converges to one as \(N\to\infty\).  Equivalently, the
probability of its complement tends to zero.  Whenever useful below, we keep
the explicit finite-\(N\) upper bound on that complementary probability
rather than using only the asymptotic phrase.

We state the external ingredients before using them.  This makes every later
application identifiable by name.

\begin{enumerate}[label=\textbf{S\arabic*.}]
\item \textbf{Markov.}  If \(X\ge0\) and \(t>0\), then
\(\Pp\{X\ge t\}\le\E X/t\).

\item \textbf{McDiarmid's inequality.}  Suppose \(F\) is a function of
independent inputs and replacing input \(i\) changes \(F\) by at most
\(c_i\).  Then, for \(t>0\),
\begin{equation}\label{eq:mcdiarmid}
 \Pp\{F-\E F\le-t\}
 \le\exp\left(-\frac{2t^2}{\sum_i c_i^2}\right).
\end{equation}

\item \textbf{Union bound.}  For arbitrary events \(A_1,\ldots,A_m\),
\(\Pp(\bigcup_j A_j)\le\sum_j\Pp(A_j)\).

\item \textbf{Subset count.}  For \(1\le k\le N\),
\begin{equation}\label{eq:subset-count}
 \binom Nk\le\frac{N^k}{k!}
 \le\left(\frac{eN}{k}\right)^k.
\end{equation}

\item \textbf{Gaussian and chi-square tails.}  If \(Z\sim\mathcal N(0,1)\),
then
\begin{equation}\label{eq:gaussian-tail}
 \Pp\{|Z|\ge t\}=2Q(t)
 \le\min\left\{2e^{-t^2/2},\frac{2\phi(t)}{t}\right\},
 \qquad t>0,
\end{equation}
where \(\phi(t)=(2\pi)^{-1/2}e^{-t^2/2}\) and
\(Q(t)=\int_t^\infty\phi(u)\,du\).  If \(X\sim\chi_N^2\) and
\(0<\varepsilon<1\), then
\begin{equation}\label{eq:chi-tail}
 \Pp\{|X-N|>\varepsilon N\}
 \le2\exp\left(-\frac{N\varepsilon^2}{8}\right).
\end{equation}

\item \textbf{Gaussian quadratic-form lower tail.}  If
\(\boldsymbol g\sim\mathcal N(\zero,\boldsymbol I_N)\) and
\(\boldsymbol 0\preceq\boldsymbol M\preceq\boldsymbol I_N\), then
\begin{equation}\label{eq:quadratic-lower-tail}
 \Pp\left\{
 \boldsymbol g^{\mathsf T}\boldsymbol M\boldsymbol g
 <\frac12\tr\boldsymbol M\right\}
 \le\exp\left(-\frac{\tr\boldsymbol M}{16}\right).
\end{equation}
Indeed, the Laurent--Massart bound gives
\(\boldsymbol g^{\mathsf T}\boldsymbol M\boldsymbol g
\ge\tr\boldsymbol M-2\sqrt{\tr(\boldsymbol M^2)t}\) except with
probability \(e^{-t}\).  Take \(t=\tr(\boldsymbol M)/16\) and use
\(\tr(\boldsymbol M^2)\le\tr\boldsymbol M\).

\item \textbf{Rank-one inverse formula.}  For invertible
\(\boldsymbol A\) satisfying
\(1+\boldsymbol u^{\mathsf T}\boldsymbol A^{-1}\boldsymbol u\ne0\),
\begin{equation}\label{eq:sherman}
 (\boldsymbol A+\boldsymbol u\boldsymbol u^{\mathsf T})^{-1}
 =\boldsymbol A^{-1}
 -\frac{\boldsymbol A^{-1}\boldsymbol u\boldsymbol u^{\mathsf T}
        \boldsymbol A^{-1}}
       {1+\boldsymbol u^{\mathsf T}\boldsymbol A^{-1}\boldsymbol u}.
\end{equation}

\end{enumerate}

Markov, McDiarmid's inequality, the union bound, and the subset estimate are
standard; see, for example, \citet{BoucheronLugosiMassart2013}.  The bound
\(Q(t)\le\phi(t)/t\) in \eqref{eq:gaussian-tail} is the classical
Mills-ratio inequality \citep{Gordon1941}; the exponential bound follows
from the Gaussian Chernoff bound.  The
chi-square estimate follows from the chi-square moment-generating function;
see \citet{LaurentMassart2000}; the same reference gives the weighted
quadratic-form lower tail used in \eqref{eq:quadratic-lower-tail}.  The
rank-one identity is due to \citet{ShermanMorrison1950}.

\section{How close rounded LMMSE gets}\label{sec:warm-start}

The descent proof will need exactly one fact from its initialization:
\begin{equation}\label{eq:warm-target}
 \dH(\boldsymbol x^{\mathrm{init}},\boldsymbol x^\star)
 \le\left\lceil\frac{N}{(\log N)^{1/4}}\right\rceil
 \quad\text{with probability tending to one.}
\end{equation}
We derive backward from this goal before proving anything about the random
matrix.

The vector \(\boldsymbol b\) in \eqref{eq:detector-lmmse} is already
normalized to estimate \(\boldsymbol x^\star\) directly.  We therefore
analyze \(\boldsymbol b\) itself.  Although the same formula is the standard
unit-covariance LMMSE estimator under a Rademacher prior, the theorem keeps
\(\boldsymbol x^\star\) deterministic; the exact identity below therefore
also uses rotational invariance of the random channel.

\subsection{What is sufficient for the desired Hamming radius?}

If coordinate \(j\) has the wrong sign, then
\(|b_j-x_j^\star|\ge1\).  Therefore, for every realization,
\begin{equation}\label{eq:signs-below-mse}
 \dH(\boldsymbol x^{\mathrm{init}},\boldsymbol x^\star)
 \le\left\|\boldsymbol b-\boldsymbol x^\star\right\|_2^2.
\end{equation}
Markov's inequality now shows what remains to be proved.  Any estimate of the
form
\begin{equation}\label{eq:mse-sufficient-goal}
 \E_{\boldsymbol H,\boldsymbol w}
 \left\|\boldsymbol b-\boldsymbol x^\star\right\|_2^2
 \le C\frac{N}{\sqrt{\log N}}
\end{equation}
with an absolute constant \(C\) implies
\[
 \Pp\left\{\dH(\boldsymbol x^{\mathrm{init}},\boldsymbol x^\star)>
 \frac{N}{(\log N)^{1/4}}\right\}
 \le\frac{C}{(\log N)^{1/4}}\longrightarrow0.
\]
Thus the LMMSE stage does not require coordinatewise accuracy.  We need only
the total expected squared error at the scale \(N/\sqrt{\log N}\).

Define the LMMSE error matrix
\begin{equation}\label{eq:B-rho}
 \boldsymbol B_\rho
 =\left(\boldsymbol I_N+
   \frac{\rho}{N}\boldsymbol H^{\mathsf T}\boldsymbol H\right)^{-1}.
\end{equation}

\subsection{The exact mean-square identity reduces the problem to a trace}

The conditional error calculation below is the standard LMMSE
error-covariance calculation; the scalar case is worked out in
\citet[Appendix~A.3]{TseViswanath2005}.  For the deterministic-word
formulation, the exact identity additionally uses right-orthogonal
invariance of the Gaussian channel.  We include both steps to make the
normalization and averaging explicit.

\begin{lemma}[Exact mean-square error of the LMMSE-form estimate]\label{lem:mse}
For every fixed \(\boldsymbol x^\star\in\{\pm1\}^N\),
\begin{equation}\label{eq:mse-identity}
 \E_{\boldsymbol H,\boldsymbol w}
 \left\|\boldsymbol b-\boldsymbol x^\star\right\|_2^2
 =\E_{\boldsymbol H}\tr(\boldsymbol B_\rho).
\end{equation}
\end{lemma}

\begin{proof}
From \eqref{eq:B-rho},
\[
 \left(\boldsymbol H^{\mathsf T}\boldsymbol H+
       \frac N\rho\boldsymbol I_N\right)^{-1}
 =\frac\rho N\boldsymbol B_\rho,
 \qquad
 \frac\rho N\boldsymbol B_\rho
 \boldsymbol H^{\mathsf T}\boldsymbol H
 =\boldsymbol I_N-\boldsymbol B_\rho.
\]
Substituting the model into the already normalized estimator
\eqref{eq:detector-lmmse} therefore gives
\[
 \boldsymbol b
 =(\boldsymbol I_N-\boldsymbol B_\rho)\boldsymbol x^\star
 +\sqrt{\frac\rho N}\,
  \boldsymbol B_\rho\boldsymbol H^{\mathsf T}\boldsymbol w.
\]
After subtracting \(\boldsymbol x^\star\),
\begin{equation}\label{eq:error-decomposition}
 \boldsymbol b-\boldsymbol x^\star
 =-\boldsymbol B_\rho\boldsymbol x^\star
 +\sqrt{\frac\rho N}\,
  \boldsymbol B_\rho\boldsymbol H^{\mathsf T}\boldsymbol w.
\end{equation}

Condition on \(\boldsymbol H\) and expand the squared norm.  The cross term
has conditional expectation zero because \(\E(\boldsymbol w)=\zero\).  The
noise term has conditional covariance
\begin{align*}
 &\frac\rho N\boldsymbol B_\rho
 \boldsymbol H^{\mathsf T}\boldsymbol H\boldsymbol B_\rho\\
 &\qquad=
 \boldsymbol B_\rho
 \left(\boldsymbol B_\rho^{-1}-\boldsymbol I_N\right)
 \boldsymbol B_\rho
 =\boldsymbol B_\rho-\boldsymbol B_\rho^2.
\end{align*}
Therefore
\begin{equation}\label{eq:mse-before-rotation}
 \E_{\boldsymbol w}\!\left[
 \left\|\boldsymbol b-\boldsymbol x^\star\right\|_2^2
 \middle|\boldsymbol H\right]
 =\boldsymbol x^{\star\mathsf T}\boldsymbol B_\rho^2\boldsymbol x^\star
  +\tr(\boldsymbol B_\rho-\boldsymbol B_\rho^2).
\end{equation}

It remains to average the first term over \(\boldsymbol H\).  For every fixed
orthogonal matrix \(\boldsymbol O\), the matrices
\(\boldsymbol H\boldsymbol O\) and \(\boldsymbol H\) have the same
distribution.  Consequently,
\(\E(\boldsymbol B_\rho^2)\) commutes with every orthogonal matrix and must be
a scalar multiple of \(\boldsymbol I_N\).  Its trace determines that scalar:
\[
 \E(\boldsymbol B_\rho^2)
 =\frac{\E\tr(\boldsymbol B_\rho^2)}{N}\boldsymbol I_N.
\]
Since \(\|\boldsymbol x^\star\|_2^2=N\),
\[
 \E\!\left[
 \boldsymbol x^{\star\mathsf T}\boldsymbol B_\rho^2\boldsymbol x^\star
 \right]
 =\E\tr(\boldsymbol B_\rho^2).
\]
Average \eqref{eq:mse-before-rotation} over \(\boldsymbol H\).  The positive
and negative copies of \(\E\tr(\boldsymbol B_\rho^2)\) cancel, leaving
\eqref{eq:mse-identity}.
\end{proof}

\paragraph{Reduction to a trace bound.}
Equations \eqref{eq:signs-below-mse} and \eqref{eq:mse-identity} reduce the
entire warm-start claim to one question:
\[
 \text{Can we prove}\qquad
 \E\tr(\boldsymbol B_\rho)=O\left(\frac{N}{\sqrt{\log N}}\right)?
\]
No coordinatewise LMMSE analysis is needed.

\subsection{From a leave-one-row identity to a trace bound}

Increasing \(\rho\) increases the positive matrix inside the inverse in
\eqref{eq:B-rho}, and therefore decreases its inverse.  Thus, for
\(\rho\ge2\log N\),
\[
 \boldsymbol B_\rho\preceq
 \left(\boldsymbol I_N+
 \frac{2\log N}{N}\boldsymbol H^{\mathsf T}\boldsymbol H\right)^{-1}.
\]
For this subsection only, write
\begin{equation}\label{eq:B-boundary}
 \boldsymbol B=\left(\boldsymbol I_N+
 \frac{2\log N}{N}\boldsymbol H^{\mathsf T}\boldsymbol H\right)^{-1}.
\end{equation}
It is enough to bound \(\E\tr(\boldsymbol B)\).

A direct row-by-row calculation with \(\boldsymbol B\) would be invalid:
row \(\boldsymbol H_{i,:}\) already appears inside \(\boldsymbol B\).  We
therefore remove that row first and define
\begin{align}
 \boldsymbol B_{-i}
 &=\left(\boldsymbol I_N+
 \frac{2\log N}{N}
 \sum_{\ell\ne i}\boldsymbol H_{\ell,:}^{\mathsf T}
                    \boldsymbol H_{\ell,:}\right)^{-1},
 \label{eq:B-minus-i}\\
 q_i
 &=\frac{2\log N}{N}\,
 \boldsymbol H_{i,:}\boldsymbol B_{-i}
 \boldsymbol H_{i,:}^{\mathsf T}.
 \label{eq:q-i}
\end{align}
Now \(\boldsymbol B_{-i}\) depends only on the other rows and is independent
of \(\boldsymbol H_{i,:}\).  The scalar \(q_i\) is not computed by the
detector.  It is the quadratic form that appears in the
Sherman--Morrison denominator when row \(i\) is restored.  Conditional on
\(\boldsymbol B_{-i}\), its unscaled quadratic part has mean
\(\tr(\boldsymbol B_{-i})\); this independence is what will connect the
random denominator to the trace we want to bound.

\begin{lemma}[The trace identity]\label{lem:trace-identity}
For every realization,
\begin{equation}\label{eq:trace-identity}
 \tr(\boldsymbol B)=\sum_{i=1}^N\frac1{1+q_i},
 \qquad
 0\le\tr(\boldsymbol B_{-i})-\tr(\boldsymbol B)\le1.
\end{equation}
More generally, when one rank-one positive-semidefinite term is added to or
removed from a matrix of the form
\(\boldsymbol I_N+\sum_j\boldsymbol u_j\boldsymbol u_j^{\mathsf T}\), the
traces of the corresponding inverses differ by at most one, regardless of
the number of terms in the sum.
\end{lemma}

\begin{proof}
Adding row \(i\) adds the rank-one matrix
\((2\log N/N)\boldsymbol H_{i,:}^{\mathsf T}\boldsymbol H_{i,:}\) to
\(\boldsymbol B_{-i}^{-1}\).  Apply the rank-one inverse formula
\eqref{eq:sherman}.  It gives
\[
 \frac{2\log N}{N}\boldsymbol H_{i,:}\boldsymbol B
 \boldsymbol H_{i,:}^{\mathsf T}
 =\frac{q_i}{1+q_i}.
\]
Now use \(\boldsymbol B^{-1}\boldsymbol B=\boldsymbol I_N\), take the trace,
and substitute the preceding equality:
\begin{align*}
 N
 &=\tr(\boldsymbol B)
 +\sum_{i=1}^N\frac{2\log N}{N}
   \boldsymbol H_{i,:}\boldsymbol B\boldsymbol H_{i,:}^{\mathsf T}\\
 &=\tr(\boldsymbol B)+\sum_{i=1}^N\frac{q_i}{1+q_i}\\
 &=\tr(\boldsymbol B)+N-\sum_{i=1}^N\frac1{1+q_i}.
\end{align*}
Cancel \(N\) to obtain the first identity.

The same rank-one formula gives
\[
 \tr(\boldsymbol B_{-i})-\tr(\boldsymbol B)
 =\frac{(2\log N/N)\boldsymbol H_{i,:}\boldsymbol B_{-i}^2
          \boldsymbol H_{i,:}^{\mathsf T}}
        {1+q_i}.
\]
Because \(\boldsymbol 0\preceq\boldsymbol B_{-i}\preceq\boldsymbol I_N\),
we have \(\boldsymbol B_{-i}^2\preceq\boldsymbol B_{-i}\).  The numerator is
therefore between zero and \(q_i\), proving that the ratio lies in
\([0,1]\).  This last calculation uses only that the inverse before the
rank-one update is bounded above by \(\boldsymbol I_N\).  It therefore proves
the stated general trace-sensitivity bound for any number of pre-existing
rank-one terms.
\end{proof}

The trace identity also suggests how to prove the correct scale.  Conditional
on \(\boldsymbol B_{-i}\), Gaussian quadratic-form concentration predicts
\[
 \boldsymbol H_{i,:}\boldsymbol B_{-i}
 \boldsymbol H_{i,:}^{\mathsf T}
 \approx\tr(\boldsymbol B_{-i}),
\]
while removing one row changes the trace by at most one.  Thus one expects
\[
 q_i\approx\frac{2\log N}{N}\tr(\boldsymbol B).
\]
If this approximation is inserted into the exact identity
\eqref{eq:trace-identity}, it yields the self-consistent relation
\[
 \tr(\boldsymbol B)
 \approx\frac{N}{1+(2\log N/N)\tr(\boldsymbol B)},
\]
whose positive solution has order \(N/\sqrt{\log N}\).  This calculation is
not used as a proof; it identifies the estimate to establish.  The next lemma
replaces the two approximations by lower-tail bounds and derives a rigorous
inequality for \(\E\tr(\boldsymbol B)\).

\begin{lemma}[Direct expected-trace bound]\label{lem:trace-bound}
For all sufficiently large \(N\),
\begin{equation}\label{eq:trace-bound}
 \E\tr(\boldsymbol B)\le\frac{2N}{\sqrt{\log N}}.
\end{equation}
\end{lemma}

\begin{proof}
Write \(m=\E\tr(\boldsymbol B)\).  We turn the heuristic preceding the
lemma into a scalar inequality for \(m\).

\textbf{Step 1: reduce the trace identity to one representative row.}
The random variables \(q_1,\ldots,q_N\) are identically distributed because
the rows of \(\boldsymbol H\) are exchangeable.  Taking expectations in
\eqref{eq:trace-identity} therefore gives
\begin{equation}\label{eq:trace-one-row}
 \frac mN
 =\E\left[\frac1{1+q_1}\right].
\end{equation}
Thus it is enough to show that \(q_1\) is usually large.  The leave-one-row
construction makes this possible because \(\boldsymbol B_{-1}\) is
independent of the Gaussian row appearing in \(q_1\).

\textbf{Step 2: the inverse built from the other rows usually has a large
trace.}
Deleting a positive-semidefinite rank-one term can only increase the inverse
in the positive-semidefinite order.  Consequently,
\[
 \E\tr(\boldsymbol B_{-1})\ge m.
\]
Replacing any one of the remaining rows changes
\(\tr(\boldsymbol B_{-1})\) by at most two: delete the old row, which changes
the trace by at most one by \eqref{eq:trace-identity}, and add the new row,
which changes it by at most one again.  Hence the bounded-difference
constants are at most two.  Applying McDiarmid's inequality at a downward
deviation of \(m/2\), and using \(\E\tr(\boldsymbol B_{-1})\ge m\), gives
\begin{equation}\label{eq:deleted-trace-prob}
 \Pp\left\{\tr(\boldsymbol B_{-1})
 <\frac m2\right\}
 \le\exp\left(
 -\frac{m^2}{8N}\right).
\end{equation}

Let
\[
 A=\left\{\tr(\boldsymbol B_{-1})\ge\frac m2\right\}.
\]

\textbf{Step 3: the independent Gaussian row produces a large quadratic
form.}
Condition on \(\boldsymbol B_{-1}\).  This matrix is now fixed, whereas
\(\boldsymbol H_{1,:}\) remains a standard Gaussian row.  Substituting
\(\boldsymbol g=\boldsymbol H_{1,:}^{\mathsf T}\) and
\(\boldsymbol M=\boldsymbol B_{-1}\) into
\eqref{eq:quadratic-lower-tail} is valid because
\(\boldsymbol 0\preceq\boldsymbol B_{-1}\preceq\boldsymbol I_N\).  Hence
\begin{align}
 &\Pp\left\{
 \boldsymbol H_{1,:}\boldsymbol B_{-1}
 \boldsymbol H_{1,:}^{\mathsf T}
 <\frac12\tr(\boldsymbol B_{-1})
 \,\middle|\,\boldsymbol B_{-1}\right\}
 \notag\\
 &\qquad\le
 \exp\left(-\frac{\tr(\boldsymbol B_{-1})}{16}\right)
 \le
 \exp\left(-\frac m{32}\right)
 \label{eq:independent-row-prob}
\end{align}
on the event \(A\), where the last inequality uses
\(\tr(\boldsymbol B_{-1})\ge m/2\).  Let
\[
 C=\left\{
 \boldsymbol H_{1,:}\boldsymbol B_{-1}
 \boldsymbol H_{1,:}^{\mathsf T}
 \ge\frac12\tr(\boldsymbol B_{-1})
 \right\}.
\]
Equation \eqref{eq:independent-row-prob} says that, conditional on any
\(\boldsymbol B_{-1}\) for which \(A\) holds, \(C\) fails with probability
at most \(e^{-m/32}\).

\textbf{Step 4: convert the two lower bounds into a reciprocal bound.}
Call \(G=A\cap C\) the good event.  On \(G\), first use \(C\), then \(A\),
in the definition of \(q_1\):
\begin{equation}\label{eq:q-direct-lower}
 \begin{aligned}
 q_1
 &=\frac{2\log N}{N}
   \boldsymbol H_{1,:}\boldsymbol B_{-1}
   \boldsymbol H_{1,:}^{\mathsf T}\\
 &\ge\frac{2\log N}{N}\cdot
       \frac12\tr(\boldsymbol B_{-1})
 \ge\frac{2\log N}{4N}m.
 \end{aligned}
\end{equation}
The function \(u\mapsto(1+u)^{-1}\) is decreasing on \([0,\infty)\).
Therefore \eqref{eq:q-direct-lower} gives, on \(G\),
\[
 \frac1{1+q_1}
 \le
 \frac{1}{1+(2\log N/4N)m}.
\]
This is the reciprocal inequality used below: it follows directly from the
preceding lower bound on \(q_1\).

\textbf{Step 5: control the complement of the good event.}
The event \(G\) requires both \(A\) and \(C\).  It therefore fails in exactly
one of the following disjoint ways:
\[
\begin{array}{ll}
\text{(i)}&
 A^c:\quad \tr(\boldsymbol B_{-1})<\dfrac m2,\\[2mm]
\text{(ii)}&
 A\cap C^c:\quad \tr(\boldsymbol B_{-1})\ge\dfrac m2
 \quad\text{but}\quad
 \boldsymbol H_{1,:}\boldsymbol B_{-1}
 \boldsymbol H_{1,:}^{\mathsf T}
 <\dfrac12\tr(\boldsymbol B_{-1}).
\end{array}
\]
The first probability is bounded by \eqref{eq:deleted-trace-prob}.  Averaging
the conditional estimate \eqref{eq:independent-row-prob} over deleted matrices in
\(A\) bounds the second.  Thus
\[
 \Pp(G^c)\le e^{-m^2/(8N)}+e^{-m/32}.
\]
Also \(q_1\ge0\), because \(\boldsymbol B_{-1}\succeq\boldsymbol0\), so
\((1+q_1)^{-1}\le1\) on \(G^c\).  Splitting the expectation in
\eqref{eq:trace-one-row} over \(G\) and \(G^c\) now yields
\begin{align}
 \frac mN
 &\le
 \frac{1}{1+(2\log N/4N)m}
 \notag\\
 &\quad+
 \exp\left(-\frac{m^2}{8N}\right)
 +\exp\left(-\frac m{32}\right).
 \label{eq:direct-trace-inequality}
\end{align}

\textbf{Step 6: solve the scalar inequality.}
We now extract \(m\le2N/\sqrt{\log N}\) from
\eqref{eq:direct-trace-inequality}.
The comparison value \(4\sqrt{N\log N}\) separates two easy cases.

If
\[
 m\le4\sqrt{N\log N},
\]
then the desired bound already holds for all sufficiently large \(N\),
because
\[
 4\sqrt{N\log N}\le\frac{2N}{\sqrt{\log N}}
 \quad\Longleftrightarrow\quad
 2\log N\le\sqrt N.
\]

It remains to consider
\[
 m>4\sqrt{N\log N}.
\]
In this case,
\[
 \exp\left(-\frac{m^2}{8N}\right)
 \le e^{-2\log N}=\frac1{N^2}.
\]
Also, for sufficiently large \(N\),
\[
 \exp\left(-\frac m{32}\right)
 \le
 \exp\left(-\frac{\sqrt{N\log N}}8\right)
 \le\frac1{N^2}.
\]
Thus the two failure contributions together are at most \(2/N^2\).  The
remaining reciprocal term is bounded using \(1/(1+u)\le1/u\).  Multiplying
the resulting inequality by \(N\) gives
\[
 m
 \le
 \frac{4N^2}{(2\log N)m}+\frac2N.
\]
Multiply by \(m\) and use \(m\le N\), which follows from
\(\boldsymbol0\preceq\boldsymbol B\preceq\boldsymbol I_N\):
\[
 m^2
 \le\frac{4N^2}{2\log N}+2.
\]
This is again at most \(4N^2/\log N\) for sufficiently large \(N\), which
is \eqref{eq:trace-bound} after taking square roots.
\end{proof}

\begin{remark}[Random-matrix interpretation of the trace scale]
The trace in \eqref{eq:mse-identity} is a linear spectral statistic of the
square Wishart matrix \(N^{-1}\boldsymbol H^{\mathsf T}\boldsymbol H\).
For \(\rho\) fixed, the Marchenko--Pastur law gives the almost-sure limit
\begin{align*}
 \frac1N\tr(\boldsymbol B_\rho)
 &\longrightarrow
 \int_0^4 \frac{1}{1+\rho x}\,
       \frac{1}{2\pi}\sqrt{\frac{4-x}{x}}\,dx \\
 &=\frac{\sqrt{1+4\rho}-1}{2\rho}.
\end{align*}
Letting \(\rho\to\infty\) after this limit gives the sharper scale
\(\rho^{-1/2}(1+o(1))\), and hence predicts
\(\tr(\boldsymbol B_\rho)\sim N/\sqrt{\rho}\)
\citep{MarchenkoPastur1967}.  Concentration of linear spectral statistics,
including Wishart examples, is treated systematically by
\citet{GuionnetZeitouni2000}.

This observation does not by itself replace the argument above at
\(\rho=\rho_N=2\log N\): the test function
\(x\mapsto(1+\rho_Nx)^{-1}\) varies with \(N\) and increasingly probes the
hard edge at zero.  A direct transfer therefore requires triangular-array
control that is not supplied merely by fixed-test-function weak convergence.
The elementary leave-one-row proof deliberately avoids that issue.  Its
constant is non-sharp, but its \(N/\sqrt{\log N}\) scale is exactly what the
downstream warm-start argument needs.

The author thanks Boaz Nadler for pointing out this spectral-statistics
interpretation, the sharper Marchenko--Pastur prediction, and the
dimension-dependent regularization caveat.
\end{remark}

The trace estimate now has exactly the scale required in
\eqref{eq:mse-sufficient-goal}.  It remains to combine it with the exact MSE
identity and the deterministic implication from a wrong sign to unit squared
error.

\subsection{Convert the trace bound into a Hamming-distance warm start}

\begin{proposition}[Rounded LMMSE enters the outer region]\label{prop:warm}
For all sufficiently large \(N\), uniformly over
\(\boldsymbol x^\star\in\{\pm1\}^N\) and \(\rho\ge2\log N\),
\begin{equation}\label{eq:warm-probability}
 \Pp\left\{\dH(\boldsymbol x^{\mathrm{init}},\boldsymbol x^\star)>
 \left\lceil\frac{N}{(\log N)^{1/4}}\right\rceil\right\}
 \le\frac{2}{(\log N)^{1/4}}
 \longrightarrow0.
\end{equation}
\end{proposition}

\begin{proof}
The exact LMMSE calculation proved
\[
 \E_{\boldsymbol H,\boldsymbol w}
 \left\|\boldsymbol b-\boldsymbol x^\star\right\|_2^2
 =\E_{\boldsymbol H}\tr(\boldsymbol B_\rho).
\]
For every \(\rho\ge2\log N\), increasing the positive matrix inside the
inverse gives
\[
 \boldsymbol B_\rho
 \preceq
 \boldsymbol B
 =
 \left(\boldsymbol I_N+
 \frac{2\log N}{N}\boldsymbol H^{\mathsf T}\boldsymbol H\right)^{-1}.
\]
The trace lemma just proved
\(\E_{\boldsymbol H}\tr(\boldsymbol B)\le2N/\sqrt{\log N}\).  Combining these three
statements yields
\[
 \E_{\boldsymbol H,\boldsymbol w}
 \left\|\boldsymbol b-\boldsymbol x^\star\right\|_2^2
 \le\frac{2N}{\sqrt{\log N}}.
\]

We also proved the following pointwise implication: every wrong rounded sign
contributes at least one to the squared error.  Hence, for every realization,
\[
 \dH(\boldsymbol x^{\mathrm{init}},\boldsymbol x^\star)
 \le
 \left\|\boldsymbol b-\boldsymbol x^\star\right\|_2^2.
\]
Markov's inequality at the threshold \(N/(\log N)^{1/4}\) now gives
\[
 \begin{aligned}
 \Pp\left\{\dH(\boldsymbol x^{\mathrm{init}},\boldsymbol x^\star)>
 \frac{N}{(\log N)^{1/4}}\right\}
 &\le
 \Pp\left\{
 \left\|\boldsymbol b-\boldsymbol x^\star\right\|_2^2>
 \frac{N}{(\log N)^{1/4}}\right\}\\
 &\le
 \frac{2N/\sqrt{\log N}}{N/(\log N)^{1/4}}
 =\frac{2}{(\log N)^{1/4}}.
 \end{aligned}
\]
The ceiling can only enlarge the allowed set.

This probability tends to zero slowly, but that causes no later loss.  In
the final theorem it is added to the failure probabilities of the uniform
landscape events.  Once all those events hold, the descent proof is
deterministic; the warm-start probability is never multiplied by the number
of iterations or used in a denominator.
\end{proof}

\begin{keybox}
\textbf{Stage 1 is complete.}
Rounded LMMSE need not recover every bit.  Its only job is to enter the ball
of radius \(\lceil N/(\log N)^{1/4}\rceil\), and the exact MSE identity plus Markov does
this directly.  The exponent \(1/4\) is deliberately loose: the same argument
gives a high-probability radius
\((N/\sqrt{\log N})g_N\) for any \(g_N\to\infty\), but a tighter radius does
not simplify or strengthen the final exact-recovery theorem.
\end{keybox}

\section{Exact algebra for objective descent}\label{sec:flip-algebra}

Rounded LMMSE has placed us inside a controlled Hamming ball.  Our remaining
goal is to show that greedy descent reaches the truth in \(O(N\log N)\)
accepted flips, before exhausting the detector's allowed number of flips.

The one-bit gain expansion is standard in LAS analyses
\citep{Sun1998,VardhanEtAl2008,Sun2009}.  We derive it in the present
normalization because the summed wrong-bit identity below is the starting
point for the uniform landscape argument.

One cannot prove this by claiming that every accepted flip corrects a wrong
bit.  The detector chooses the largest objective decrease among all
coordinates, and that coordinate need not always be one of the currently
wrong coordinates.  Hamming distance may therefore increase at an
intermediate step.  The quantity that decreases monotonically is the
least-squares objective \(f\).

The descent proof consequently needs two facts at every nontruth vector in
the controlled ball:
\begin{enumerate}[label=\textbf{\arabic*.}]
\item The current objective is above the objective at the truth.  Upper and
lower bounds on this objective gap will confine the path and give a total
budget for the number of flips.
\item At least one wrong bit offers a definite objective decrease.  The
greedy detector chooses the best flip among all bits, so its chosen decrease
is at least as large.
\end{enumerate}
This section derives the exact identities from which
\Cref{sec:outer-region,sec:inner-region} will prove these two facts.

\subsection{Normalize the truth and encode the current errors}

The channel distribution is invariant under column sign changes.  We first
use this symmetry to normalize the transmitted word, and then represent any
current iterate solely by the set of coordinates on which it is wrong.  This
notation makes the two quantities needed by descent explicit.

Multiplying column \(i\) of \(\boldsymbol H\) by \(x_i^\star\) preserves its
standard Gaussian distribution.  It maps the transmitted vector to
\(\one\), preserves Hamming distance, and maps every bit-flip gain to the
corresponding gain in the transformed problem.  The LMMSE initialization
transforms in the same way almost surely.  Indeed, the unrounded LMMSE
coordinates have continuous distributions, so the convention used for
\(\sign(0)\) is encountered with probability zero.  We may therefore prove
the remaining statements for
\[
 \boldsymbol x^\star=\one.
\]

For a current vector \(\boldsymbol x\), let
\begin{equation}\label{eq:error-set}
 S=\{i:x_i=-1\}
\end{equation}
be its set of wrong coordinates, and write \(\boldsymbol h_i\) for column
\(i\) of \(\boldsymbol H\).  Every coordinate in \(S\) equals \(-1\) instead
of \(+1\), so its error has size two.  Therefore
\[
 \boldsymbol H(\one-\boldsymbol x)
 =2\sum_{i\in S}\boldsymbol h_i.
\]
Substituting the channel model gives the current residual:
\begin{equation}\label{eq:residual}
 \begin{aligned}
 \boldsymbol y-\sqrt{\frac\rho N}\boldsymbol H\boldsymbol x
 &=
 \boldsymbol w+
 \sqrt{\frac\rho N}\boldsymbol H(\one-\boldsymbol x)\\
 &=
 \boldsymbol w+2\sqrt{\frac\rho N}
 \sum_{i\in S}\boldsymbol h_i.
 \end{aligned}
\end{equation}
It is the original noise plus the channel contribution of the wrong bits.

\subsection{First question: how far above the truth is the objective?}

At the truth, the residual is \(\boldsymbol w\), and hence
\[
 f(\one)=\|\boldsymbol w\|_2^2.
\]
The following identity compares the current objective with that value.

\begin{lemma}[Exact objective gap]\label{lem:objective-gap}
For every error set \(S\),
\begin{equation}\label{eq:gap-identity}
 f(\boldsymbol x)-f(\one)
 =\frac{4\rho}{N}
 \left\|\sum_{i\in S}\boldsymbol h_i\right\|_2^2
+4\sqrt{\frac\rho N}\,
 \boldsymbol w^{\mathsf T}\sum_{i\in S}\boldsymbol h_i.
\end{equation}
\end{lemma}

\begin{proof}
Insert the residual from \eqref{eq:residual}, expand the square, and subtract
the objective at the truth:
\begin{align*}
 f(\boldsymbol x)-f(\one)
 &=
 \left\|\boldsymbol w+2\sqrt{\frac\rho N}
 \sum_{i\in S}\boldsymbol h_i\right\|_2^2
 -\|\boldsymbol w\|_2^2\\
 &=
 \frac{4\rho}{N}
 \left\|\sum_{i\in S}\boldsymbol h_i\right\|_2^2
+4\sqrt{\frac\rho N}\,
 \boldsymbol w^{\mathsf T}\sum_{i\in S}\boldsymbol h_i.
\end{align*}
\end{proof}

The first term is the positive signal penalty caused by the wrong bits.  For
a typical Gaussian channel it is approximately \(4\rho|S|\).  The second
term is the projection of the noise onto the same sum of channel columns; it
can have either sign.  The later probability lemmas will show uniformly that
the noise term cannot cancel the signal penalty.  That will place every
controlled nontruth vector strictly above the truth.

\subsection{Second question: does some wrong bit offer progress?}

For \(i\in S\), let \(\boldsymbol x^{(i)}\) be the vector obtained by
correcting wrong bit \(i\).  Correcting it subtracts
\(2\sqrt{\rho/N}\boldsymbol h_i\) from the residual.  Expanding this single
change gives its exact objective improvement:
\begin{equation}\label{eq:one-wrong-gain}
\begin{aligned}
 f(\boldsymbol x)-f(\boldsymbol x^{(i)})
 &=
 4\sqrt{\frac\rho N}\boldsymbol h_i^{\mathsf T}\boldsymbol w
 +\frac{8\rho}{N}\boldsymbol h_i^{\mathsf T}
   \sum_{j\in S}\boldsymbol h_j
 -\frac{4\rho}{N}\|\boldsymbol h_i\|_2^2.
\end{aligned}
\end{equation}

Controlling every expression in \eqref{eq:one-wrong-gain} separately would
be unnecessarily difficult.  Instead, sum the gains offered by all wrong
bits.  The cross terms then collapse into one squared norm.

\begin{lemma}[Exact total gain offered by the wrong bits]
\label{lem:total-wrong-gain}
For every nonempty error set \(S\),
\begin{equation}\label{eq:gain-identity}
\begin{aligned}
 &\sum_{i\in S}
 \bigl[f(\boldsymbol x)-f(\boldsymbol x^{(i)})\bigr]\\
 &\quad=4\Bigg[
 \frac\rho N\left(
 2\left\|\sum_{i\in S}\boldsymbol h_i\right\|_2^2
 -\sum_{i\in S}\|\boldsymbol h_i\|_2^2\right)
 +\sqrt{\frac\rho N}\,
 \boldsymbol w^{\mathsf T}\sum_{i\in S}\boldsymbol h_i
 \Bigg].
\end{aligned}
\end{equation}
\end{lemma}

\begin{proof}
Sum the individual identity \eqref{eq:one-wrong-gain} over \(i\in S\).
The noise terms sum directly.  For the middle terms,
\[
 \sum_{i\in S}\boldsymbol h_i^{\mathsf T}
 \sum_{j\in S}\boldsymbol h_j
 =
 \left(\sum_{i\in S}\boldsymbol h_i\right)^{\mathsf T}
 \left(\sum_{j\in S}\boldsymbol h_j\right)
 =
 \left\|\sum_{i\in S}\boldsymbol h_i\right\|_2^2.
\]
Substitution gives \eqref{eq:gain-identity}.
\end{proof}

\begin{lemma}[From total gain to one profitable wrong bit]
\label{lem:average-wrong-gain}
If the right side of \eqref{eq:gain-identity} is at least \(a|S|\), then
some wrong bit \(i\in S\) satisfies
\[
 f(\boldsymbol x)-f(\boldsymbol x^{(i)})\ge a.
\]
The maximal available gain is therefore also at least \(a\).  If
\(a\ge1/N\), the detector accepts a flip with gain at least \(a\).
\end{lemma}

\begin{proof}
The left side of \eqref{eq:gain-identity} is a sum of \(|S|\) wrong-bit
gains.  At least one summand is at least their average.  The detector
maximizes the gain over all \(N\) bits, so its chosen gain cannot be smaller
than that particular wrong-bit gain.
\end{proof}

\begin{keybox}
\textbf{Exact-algebra stage complete.}
The two exact identities depend only on the length of
\(\sum_{i\in S}\boldsymbol h_i\), the individual column lengths, and the
noise projection onto that column sum.
\Cref{sec:outer-region,sec:inner-region} control these quantities uniformly
over all relevant error sets.  In the outer region they
will produce a wrong-bit gain of at least \(\rho/2\); in the inner region,
at least \(4\rho/(\log N)^{5/4}\).  \Cref{sec:descent} will use the objective
gap as a barrier and as the finite budget that limits the total number of
flips.
\end{keybox}

\section{The outer region: constant-factor bounds are enough}\label{sec:outer-region}

The rounded LMMSE estimate begins with at most
\(\lceil N/(\log N)^{1/4}\rceil\) wrong bits.  Greedy descent need not reduce Hamming
distance at every step, so we prove the landscape estimates on the larger
ball extending to \(\lceil4N/(\log N)^{1/4}\rceil\) wrong bits.  The factor four is
chosen later to create an objective barrier; it is not a statistical
threshold.

We split this ball into two cardinality ranges.  This section treats
\[
 \left\lceil\frac{N}{(\log N)^3}\right\rceil
 <|S|\le
 \left\lceil\frac{4N}{(\log N)^{1/4}}\right\rceil.
\]
We call this the \emph{outer region}.  The word ``outer'' refers only to
Hamming distance from the transmitted word.

For a fixed number \(k\) of wrong bits, a binary vector is determined by
choosing which \(k\) coordinates are wrong.  Therefore the number of vectors
at Hamming distance exactly \(k\) from the truth is
\[
 \binom Nk.
\]
For every \(1\le k\le N\), the standard counting inequality gives
\[
 \binom Nk\le\left(\frac{eN}{k}\right)^k,
 \qquad
 \log\binom Nk\le k\log\frac{eN}{k}.
\]
Here ``the logarithm of the number of sets'' means the logarithm of the
integer \(\binom Nk\); it is not a logarithm applied to a set.  The relevant
sets are precisely the possible error sets \(S\) in the Hamming range under
discussion.

In the outer region, \(k>N/(\log N)^3\), so
\[
 \frac{N}{k}< (\log N)^3.
\]
Consequently,
\[
 \log\binom Nk
 \le k\log\frac{eN}{k}
 \le k\log\!\bigl(e(\log N)^3\bigr)
 =k(1+3\log\log N).
\]
Equivalently, the number of vectors at distance \(k\) is at most
\[
 \exp\{k(1+3\log\log N)\}.
\]
This is the origin of the expression
\(e^{O(k\log\log N)}\).  There are at most \(N\) possible integer values of
\(k\), because \(k\in\{1,\ldots,N\}\).  This elementary observation is what
later produces a multiplicative factor \(N\) when we sum bounds over all
cardinalities.

Two deliberately coarse estimates suffice here.  First, the squared norm of
the erroneous channel sum
\(\|\sum_{i\in S}\boldsymbol h_i\|_2^2\) is within fixed constant factors of
its mean \(N|S|\).  For one fixed set this estimate fails with probability
\(e^{-cN}\), and there are only \(\exp\{o(N)\}\) sets in the controlled ball.
Second, for one fixed set the probability that its noise projection is too
large is \(e^{-c|S|\log N}\).  There are only
\(e^{O(|S|\log\log N)}\) sets of that size in this region.  Thus both union
bounds have ample slack.

\paragraph{Why two concentration regimes are needed.}
The sharper channel estimate used in \Cref{sec:inner-region} has relative
error \((\log N)^{-5/4}\).  Instead of merely
requiring the squared channel-sum length to lie between fixed multiples of
its mean \(N|S|\), we require it to lie in
\[
 \left[1-\frac1{(\log N)^{5/4}},
       1+\frac1{(\log N)^{5/4}}\right]N|S|.
\]
For one fixed set \(S\), the probability of violating this interval is at
most
\[
 2\exp\left\{-\frac{N}{8(\log N)^{5/2}}\right\}.
\]
The positive quantity \(N/[8(\log N)^{5/2}]\) appearing after the minus sign
is commonly called the concentration exponent or failure exponent.  We use
the phrase only as shorthand for this displayed probability bound.

If the inner region ended at
\(N/(\log N)^a\), the logarithm of the number of relevant sets would be of
order \(N\log\log N/(\log N)^a\).  The latter is negligible compared with
the former exactly when \(a>5/2\).  We use the simple nearby choice \(a=3\).
The split is not intrinsic: any fixed \(a>5/2\) would work.  The choice
\(a=5/2\) misses by a \(\log\log N\) factor; a value such as \(a=11/2\)
would work but would make the sharp inner region unnecessarily small.

\begin{lemma}[Uniform length of erroneous channel sums]\label{lem:outer-channel}
With probability tending to one, simultaneously for every nonempty set with
\[
 |S|\le\left\lceil\frac{4N}{(\log N)^{1/4}}\right\rceil,
\]
we have
\begin{equation}\label{eq:outer-channel}
 \frac34N|S|
 \le\left\|\sum_{i\in S}\boldsymbol h_i\right\|_2^2
 \le\frac54N|S|.
\end{equation}
In particular, \(\|\boldsymbol h_i\|_2^2\le5N/4\) for every column.
\end{lemma}

\begin{proof}
Fix one nonempty set \(S\).  Since the columns are independent standard
Gaussians,
\[
 \sum_{i\in S}\boldsymbol h_i
 \sim\mathcal N(\boldsymbol 0,|S|\boldsymbol I_N).
\]
It follows that
\[
 \frac1{|S|}\left\|\sum_{i\in S}\boldsymbol h_i\right\|_2^2
 \sim\chi_N^2.
\]
The chi-square bound with relative error \(1/4\) says that this one fixed
set violates the interval
\([3N/4,5N/4]\) with probability at most \(2e^{-N/128}\).

We now count all sets for which this bound must hold.  Put
\[
 m=\left\lceil\frac{4N}{(\log N)^{1/4}}\right\rceil.
\]
This abbreviation is used only in the present counting paragraph.  Since
\(k\mapsto k\log(eN/k)\) is increasing for \(k<N\), for every allowed \(k\),
\[
 \binom Nk
 \le \exp\left\{
 m\log\left(\frac{eN}{m}\right)
 \right\}
\]
for all sufficiently large \(N\), when \(m<N\).  There are at most \(N\) allowed values of
\(k\).  Therefore the total number of nonempty sets in the stated range is at
most \(N\) times the preceding display.  Taking its logarithm gives
\[
 O\left(
 \frac{N\log\log N}{(\log N)^{1/4}}
 \right)=o(N).
\]
The union bound over all those sets is consequently at most
\[
 2\exp\left\{
 O\left(\frac{N\log\log N}{(\log N)^{1/4}}\right)
 -\frac{N}{128}
 \right\},
\]
which tends to zero because
\(N\log\log N/(\log N)^{1/4}=o(N)\).

Finally take a singleton set \(S=\{i\}\).  Then
\(\sum_{j\in S}\boldsymbol h_j=\boldsymbol h_i\).  The upper half of the
same uniform bound gives
\[
 \|\boldsymbol h_i\|_2^2\le\frac54N
\]
simultaneously for every column \(i\).  This is the ``individual-column
statement'' used later when bounding
\(\sum_{i\in S}\|\boldsymbol h_i\|_2^2\).
\end{proof}

\begin{lemma}[Uniform control of outer-region noise projections]\label{lem:outer-noise}
There is one event, depending only on \(\boldsymbol H\) and \(\boldsymbol w\),
whose probability tends to one and on which the following inequality holds
for every \(\rho\ge2\log N\) and every set in the outer region:
\begin{equation}\label{eq:outer-noise}
 \left|\sqrt{\frac\rho N}\,
 \boldsymbol w^{\mathsf T}\sum_{i\in S}\boldsymbol h_i\right|
 \le\frac{\rho|S|}{8}.
\end{equation}
\end{lemma}

\begin{proof}
Condition on a channel for which every erroneous column sum has the length
asserted in the preceding lemma.  Fix one set \(S\) in the outer region.  The
vector \(\sum_{i\in S}\boldsymbol h_i\) is now fixed, while
\(\boldsymbol w\) is still standard Gaussian.  Therefore
\[
 \sqrt{\frac\rho N}\,
 \boldsymbol w^{\mathsf T}\sum_{i\in S}\boldsymbol h_i
\]
is a one-dimensional centered Gaussian.  Its variance is
\[
 \frac\rho N
 \left\|\sum_{i\in S}\boldsymbol h_i\right\|_2^2,
\]
which is at most
\[
 \frac\rho N\cdot\frac54N|S|=\frac54\rho|S|.
\]
Thus its typical magnitude is of order \(\sqrt{\rho|S|}\).  We need it to
stay below \(\rho|S|/8\), which is of order \(\sqrt{\rho|S|}\) standard
deviations from zero.  Applying
\(\Pp\{|Z|\ge t\}\le2e^{-t^2/2}\) to this one fixed set gives
\begin{equation}\label{eq:outer-one-set}
 \Pp\left\{\left|\sqrt{\frac\rho N}\,
 \boldsymbol w^{\mathsf T}\sum_{i\in S}\boldsymbol h_i\right|
 >\frac{\rho|S|}{8}\ \middle|\ \boldsymbol H\right\}
 \le2\exp\left(-\frac{\rho|S|}{160}\right).
\end{equation}

We now ask whether this estimate holds for every set simultaneously.  At the
smallest allowed SNR, \(\rho=2\log N\), the failure exponent for one fixed
set is \(|S|\log N/80\).  Here ``one fixed set'' refers to the conditional
failure probability in the preceding display for one specified \(S\), before
we union-bound over the possible choices of \(S\).

The number of sets having this cardinality is \(\binom N{|S|}\).  Because
every outer-region cardinality satisfies
\(|S|>N/(\log N)^3\), the counting calculation given at the start of this
section gives
\[
 \log\binom N{|S|}
 \le |S|\log\frac{eN}{|S|}
 \le |S|(1+3\log\log N).
\]
The phrase ``for all sufficiently large \(N\)'' means that there exists a
fixed finite integer \(N_0\) such that the assertion holds for every
\(N\ge N_0\).  It is needed here because the asymptotic fact
\(\log\log N=o(\log N)\) does not imply the displayed numerical inequality
for every small or moderate \(N\).  For \(N\ge N_0\),
\[
 1+3\log\log N<\frac{\log N}{160}.
\]
Therefore, for one fixed cardinality \(k\) in the outer region,
\begin{align*}
 &\binom Nk\,2\exp\left(-\frac{k\log N}{80}\right)\\
 &\quad\le
 2\exp\left\{k(1+3\log\log N)-\frac{k\log N}{80}\right\}\\
 &\quad\le2\exp\left(-\frac{k\log N}{160}\right).
\end{align*}
There are fewer than \(N\) possible cardinalities because \(k\) is an integer
between one and \(N\); this does not mean that their number is \(o(N)\).
Summing the last bound over those cardinalities and using
\(k>N/(\log N)^3\) gives
\[
 2N\exp\left(
 -\frac{N\log N}{160(\log N)^3}
 \right)=o(1).
\]
This tends to zero.  Increasing \(\rho\) makes the required threshold larger
in standard-deviation units: after dividing both sides of the desired bound
by \(\sqrt\rho\), its right-hand side grows as \(\sqrt\rho\).  Hence the
event proved at \(\rho=2\log N\) automatically implies the corresponding
event for every larger \(\rho\) on the same realization.  This is the precise
meaning of ``uniformly over \(\rho\ge2\log N\).''

If the channel-sum event from \Cref{lem:outer-channel} is denoted in words by ``all coarse
channel bounds hold,'' the total failure probability in this lemma is at most
\[
 \Pp\{\text{some coarse channel bound fails}\}
 +2N\exp\left(-\frac{N\log N}{160(\log N)^3}\right).
\]
The first probability tends to zero by \Cref{lem:outer-channel}, and the second displayed
quantity tends to zero by direct inspection.  This supplies an explicit
probability bound rather than only the phrase ``with probability tending to
one.''
\end{proof}

\begin{lemma}[Excess cost and one-step progress in the outer region]
\label{lem:outer-gain}
Suppose the following two statements hold simultaneously for every relevant
error set: the squared channel-sum norm lies between
\((3/4)N|S|\) and \((5/4)N|S|\), and the magnitude of the scaled noise
projection is at most \(\rho|S|/8\).  The preceding two lemmas show that this
joint event has probability tending to one.  On this event, consider any
binary vector whose error set satisfies
\[
 \left\lceil\frac{N}{(\log N)^3}\right\rceil
 <|S|\le
 \left\lceil\frac{4N}{(\log N)^{1/4}}\right\rceil.
\]
Then
\begin{equation}\label{eq:outer-gap}
 \frac52\rho|S|
 \le f(\boldsymbol x)-f(\one)
 \le\frac{11}{2}\rho|S|,
\end{equation}
and
\begin{equation}\label{eq:outer-total-gain}
 \sum_{i\in S}\bigl[f(\boldsymbol x)-f(\boldsymbol x^{(i)})\bigr]
 \ge\frac12\rho|S|.
\end{equation}
Hence correcting at least one currently wrong bit decreases the cost by at
least \(\rho/2\).
\end{lemma}

\begin{proof}
The exact objective-gap identity \eqref{eq:gap-identity} is
\[
 f(\boldsymbol x)-f(\one)
 =\frac{4\rho}{N}
 \left\|\sum_{i\in S}\boldsymbol h_i\right\|_2^2
 +4\sqrt{\frac\rho N}\,
 \boldsymbol w^{\mathsf T}\sum_{i\in S}\boldsymbol h_i.
\]
The channel-sum norm lies between \((3/4)N|S|\) and
\((5/4)N|S|\), and the scaled noise projection has magnitude at most
\(\rho|S|/8\).  The smallest possible coefficient of \(\rho|S|\) is
\[
 4\left(\frac34-\frac18\right)=\frac52,
\]
and the largest possible coefficient is
\(4(5/4+1/8)=11/2\).  This proves the displayed objective bounds.

For progress, recall the exact identity for the sum of all gains offered by
the currently wrong bits:
\[
4\left[
 \frac\rho N\left(
 2\left\|\sum_{i\in S}\boldsymbol h_i\right\|_2^2
 -\sum_{i\in S}\|\boldsymbol h_i\|_2^2\right)
 +\sqrt{\frac\rho N}\,
 \boldsymbol w^{\mathsf T}\sum_{i\in S}\boldsymbol h_i
\right].
\]
The gain identity contains three contributions inside its outer factor of
four: twice the squared channel-sum norm, minus the sum of the individual
column norms, plus the noise projection.  To obtain a lower bound, use the
lower channel-sum bound, the upper singleton bound
\(\sum_{i\in S}\|\boldsymbol h_i\|_2^2\le(5/4)N|S|\), and the most negative
allowed noise projection.  Substituting these three bounds directly into
\eqref{eq:gain-identity} gives
\[
 \begin{aligned}
 \sum_{i\in S}\Delta_i(\boldsymbol x)
 &\ge4\rho|S|\left(
    2\cdot\frac34-\frac54-\frac18\right)\\
 &=\frac12\rho|S|.
 \end{aligned}
\]
Thus the sum of the cost decreases obtained by correcting the \(|S|\)
currently wrong bits is at least
\((\rho/2)|S|\).  At least one of those gains is therefore at least their
average, \(\rho/2\).  Since the greedy detector chooses the largest cost
decrease among
all \(N\) coordinates, its next accepted step decreases the objective by at
least \(\rho/2\).  This statement guarantees objective progress; it does not
claim that every bit selected by the greedy rule must currently be wrong.
\end{proof}

\begin{keybox}
\textbf{Outer-region stage complete.}
Uniform constant-factor channel and noise bounds place every outer-region
vector strictly above the truth and certify a wrong-bit gain of at least
\(\rho/2\).  The lower objective bound at the outer boundary will become the
barrier that keeps the descent path inside the controlled ball.
\end{keybox}

\section{The inner region: sharp control near the truth}\label{sec:inner-region}

We now treat every nonempty error set satisfying
\[
 1\le |S|\le\left\lceil\frac{N}{(\log N)^3}\right\rceil.
\]
This is the \emph{inner region}.  It meets the outer region without a gap:
the boundary cardinality belongs to the inner region, and the next larger
cardinality belongs to the outer region.

The terminology can otherwise be misleading.  The inner region contains the
vectors \emph{closer} to the truth and therefore has the smaller error
cardinalities.  Its concentration conclusion is tighter, not weaker: the
channel norm must have relative error only \((\log N)^{-5/4}\).  Such a tight
statement can hold uniformly here precisely because the inner region contains
far fewer candidate error sets than the full controlled ball.

The exact numbers of candidate binary vectors are
\[
 \sum_{k=1}^{\lceil N/(\log N)^3\rceil}\binom Nk
 \quad\text{in the inner region}
\]
and
\[
 \sum_{k=\lceil N/(\log N)^3\rceil+1}^{
             \lceil4N/(\log N)^{1/4}\rceil}
 \binom Nk
 \quad\text{in the outer region}.
\]
Their logarithms are bounded respectively by
\[
 O\left(\frac{N\log\log N}{(\log N)^3}\right)
 \quad\text{and}\quad
 O\left(\frac{N\log\log N}{(\log N)^{1/4}}\right).
\]
The outer region has many more vectors; this is why it supports only the
coarser uniform channel estimate.

Before proving concentration inequalities, we identify the exact Gaussian
tail that matters.  Suppose temporarily that \(|S|=k\) and that the erroneous
channel sum has its typical squared length,
\[
 \left\|\sum_{i\in S}\boldsymbol h_i\right\|_2^2\approx Nk.
\]
The objective gap is
\[
 f(\boldsymbol x)-f(\one)
 =\frac{4\rho}{N}
 \left\|\sum_{i\in S}\boldsymbol h_i\right\|_2^2
 +4\sqrt{\frac\rho N}\,
 \boldsymbol w^{\mathsf T}\sum_{i\in S}\boldsymbol h_i.
\]
The positive signal term is therefore approximately \(4\rho k\).  For the
noise to cancel it, one would need approximately
\[
 \boldsymbol w^{\mathsf T}\sum_{i\in S}\boldsymbol h_i
 \le-\sqrt{\rho N}\,k.
\]
Conditional on the channel, the random variable on the left is Gaussian with
standard deviation approximately \(\sqrt{Nk}\).  Dividing the required
negative fluctuation by this standard deviation gives
\[
 \frac{-\sqrt{\rho N}\,k}{\sqrt{Nk}}=-\sqrt{\rho k}.
\]
Thus the dangerous one-set probability is approximately
\[
 Q(\sqrt{\rho k}).
\]
Here ``dangerous one-set probability'' means the probability, for one fixed
error set \(S\), that the negative noise projection is large enough to nearly
cancel the positive signal contribution to the excess cost.  If full
cancellation occurred, that particular nontruth vector could have cost no
larger than the truth.  The rigorous lemma controls both signs of the noise
because the later gain calculation also needs an upper-magnitude bound.

This calculation already reveals the constant two.  If
\(\rho=c\log N\), the Gaussian tail is, up to its prefactor,
\[
 Q(\sqrt{c k\log N})
 \approx\frac{N^{-ck/2}}{\sqrt{k\log N}}.
\]
There are approximately \(N^k/k!\) error sets of size \(k\).  Their total
union-bound contribution is therefore approximately
\[
 \frac{N^{k(1-c/2)}}{k!\sqrt{k\log N}}.
\]
For \(c>2\), a negative power of \(N\) remains.  For \(c=2\), the powers of
\(N\) cancel, but the Gaussian prefactor and \(1/k!\) still make the sum tend
to zero.  For \(c<2\), the power of \(N\) has the wrong sign; already the
one-bit contribution grows instead of vanishing.  The rigorous proof below adds a
small relative slack to this calculation but changes none of its leading
terms.

Constant-factor concentration would perturb the exponent \(ck/2\) by a
constant factor and would therefore lose the boundary \(c=2\).  We instead
keep relative error only \((\log N)^{-5/4}\).  For sets of size \(k\), this
changes the exponent by \(9k/(\log N)^{1/4}\).  The correction per unit
\(k\) tends to zero; the total correction need not.  Equivalently, it
contributes only an \(e^{o(k)}\) factor, which the factorial in the subset
count absorbs.

The achievability proof does not require a global comparison with all
\(2^N-1\) nontruth words.  It requires two properties only inside the
controlled Hamming ball: every nontruth vector there must lie above the
truth, and some correcting flip must lower its cost.  In
\Cref{sec:descent}, the objective barrier will prove that the algorithm never
leaves this ball, so no landscape information outside it is needed.

\paragraph{Why vanishing relative error is needed, and why the exponent
\(5/4\) is convenient.}
Suppose more generally that the relative-error level were \((\log N)^{-b}\).
Preserving the sharp Gaussian exponent without a constant loss is automatic
for \(b>1\).  The endpoint \(b=1\) leaves only an \(e^{Ck}\) factor, which
is still summable against \(k!\).  On the other hand, the
fixed-set chi-square exponent would be
\(N/(\log N)^{2b}\).  To union-bound over the inner sets, this must dominate
their counting exponent \(N\log\log N/(\log N)^3\), which holds when
\(b<3/2\).  Any fixed choice
\[
 1\le b<\frac32
\]
works.  The exponent \(5/4\) is the midpoint of this admissible interval and
keeps the later arithmetic simple; it is not a fundamental constant.

\subsection{Establish uniform sharp channel geometry}

The threshold calculation requires the squared channel-sum norm to remain
within a vanishing relative error of its mean \(N|S|\).  We first prove that
this sharper channel event holds simultaneously for every inner error set;
the following subsection will then condition on it to control the noise.

\begin{lemma}[Sharp concentration of the erroneous channel sum]
\label{lem:inner-channel}
For all sufficiently large \(N\), with probability tending to one,
simultaneously for every nonempty set with
\[
 |S|\le\left\lceil\frac{N}{(\log N)^3}\right\rceil,
\]
we have
\begin{equation}\label{eq:inner-channel}
 \left(1-\frac1{(\log N)^{5/4}}\right)N|S|
 \le\left\|\sum_{i\in S}\boldsymbol h_i\right\|_2^2
 \le
 \left(1+\frac1{(\log N)^{5/4}}\right)N|S|.
\end{equation}
\end{lemma}

\begin{proof}
Fix a set of size \(|S|\).  The sum of its independent Gaussian columns has
distribution \(\mathcal N(\zero,|S|\boldsymbol I_N)\).  Therefore
\[
 \frac1{|S|}\left\|\sum_{i\in S}\boldsymbol h_i\right\|_2^2
 \sim\chi_N^2.
\]
Apply \eqref{eq:chi-tail} with
\(\varepsilon=(\log N)^{-5/4}\).  The failure probability for one fixed set
is at most
\[
 2\exp\left(-\frac{N}{8(\log N)^{5/2}}\right).
\]

We next count the sets over which this fixed-set estimate must be union-bounded.
For every integer
\(1\le k\le\lceil N/(\log N)^3\rceil\),
\[
 \binom Nk\le\exp\left\{k\log\frac{eN}{k}\right\}.
\]
The exponent \(k\log(eN/k)\) increases with \(k\) for \(k<N\).  Hence every
term in the sum over cardinalities is at most the value obtained at
\(k=\lceil N/(\log N)^3\rceil\).  There are at most
\(\lceil N/(\log N)^3\rceil\) terms.  Therefore the total number of nonempty
inner-region error sets is at most
\[
 \left\lceil\frac{N}{(\log N)^3}\right\rceil
 \exp\left\{
 \left\lceil\frac{N}{(\log N)^3}\right\rceil
 \log\left(
 \frac{eN}{\lceil N/(\log N)^3\rceil}
 \right)\right\}.
\]
Ignoring only ceiling corrections, the logarithm inside the exponential is
\[
 \log\left(\frac{eN}{N/(\log N)^3}\right)
 =1+3\log\log N.
\]
It follows that the logarithm of the total number of inner-region sets is
\[
 O\left(\frac{N\log\log N}{(\log N)^3}\right).
\]
The union bound is the number of sets multiplied by the fixed-set failure
probability.  In particular, the total failure probability is at most
\[
 2\left\lceil\frac{N}{(\log N)^3}\right\rceil
 \exp\left\{
 \left\lceil\frac{N}{(\log N)^3}\right\rceil
 \log\left(
 \frac{eN}{\lceil N/(\log N)^3\rceil}
 \right)
 -\frac{N}{8(\log N)^{5/2}}
 \right\}.
\]
To see directly that this expression tends to zero, divide its positive
counting exponent by the negative fixed-set concentration exponent
\(N/[8(\log N)^{5/2}]\).  The ratio is
\[
 O\left(\frac{\log\log N}{\sqrt{\log N}}\right)\longrightarrow0.
\]
Thus the negative concentration exponent dominates the logarithm of the
number of sets, and the total failure probability tends to zero.
\end{proof}

\subsection{Control all inner noise projections at the threshold scale}

The preceding lemma fixes the conditional variance of every relevant noise
projection up to the required relative error.  We now use the resulting
Gaussian tail at \(\rho=2\log N\), followed by a union bound over the inner
error sets.  This is the step at which the leading constant \(2\) matters.

\begin{lemma}[Sharp uniform bound on the noise projection]
\label{lem:inner-noise}
For all sufficiently large \(N\), there is one event, depending only on
\(\boldsymbol H\) and \(\boldsymbol w\),
whose probability tends to one and on which the following inequality holds
for every \(\rho\ge2\log N\) and every nonempty set satisfying
\[
 |S|\le\left\lceil\frac{N}{(\log N)^3}\right\rceil,
\]:
\begin{equation}\label{eq:inner-noise}
 \left|\boldsymbol w^{\mathsf T}
 \sum_{i\in S}\boldsymbol h_i\right|
 \le
 \left(1-\frac4{(\log N)^{5/4}}\right)
 \sqrt{\rho N}\,|S|.
\end{equation}
\end{lemma}

\begin{proof}
It is enough to prove the statement at \(\rho=2\log N\), because its
right-hand side increases with \(\rho\).  Thus the event at the boundary SNR
is contained in the corresponding event at every larger SNR, which is again
the precise meaning of uniformity over \(\rho\).

Condition on a channel satisfying \eqref{eq:inner-channel}.  For one fixed
set, the projection is conditionally Gaussian:
\[
 \left.\boldsymbol w^{\mathsf T}\sum_{i\in S}\boldsymbol h_i
 \ \middle|\ \boldsymbol H\right.
 \sim\mathcal N\left(
 0,\left\|\sum_{i\in S}\boldsymbol h_i\right\|_2^2\right).
\]
The forbidden magnitude in the lemma is
\[
 \left(1-\frac4{(\log N)^{5/4}}\right)
 \sqrt{2N\log N}\,|S|.
\]
The conditional standard deviation is the norm of the channel sum.  The
preceding lemma bounds that standard deviation by
\[
 \sqrt{\left(1+\frac1{(\log N)^{5/4}}\right)N|S|}.
\]
Dividing the forbidden magnitude by this upper bound produces
\[
 \sqrt{
 2|S|\log N\,
 \frac{(1-4/(\log N)^{5/4})^2}
      {1+1/(\log N)^{5/4}}}.
\]
This is the rigorous version of the heuristic threshold
\(\sqrt{\rho|S|}\).  The appearance of \(\log N\) is not an additional
assumption or a new concentration slack: we have simply substituted the
boundary value \(\rho=2\log N\), so
\(\sqrt{\rho|S|}=\sqrt{2|S|\log N}\).  The extra fraction under the square
root is the small relative slack created by the channel relative-error bound and by
placing the permitted noise magnitude slightly below complete cancellation.
The conditional probability of violating the desired
two-sided bound is therefore at most
\begin{equation}\label{eq:inner-q}
 2Q\left(\sqrt{
 2|S|\log N\,
 \frac{(1-4/(\log N)^{5/4})^2}
      {1+1/(\log N)^{5/4}}}
 \right).
\end{equation}
We now quantify that small relative slack.  The fraction inside the
square root satisfies
\begin{align*}
 \frac{(1-4/(\log N)^{5/4})^2}
      {1+1/(\log N)^{5/4}}
 &=1-\frac9{(\log N)^{5/4}}
 +\frac{25/(\log N)^{5/2}}
       {1+1/(\log N)^{5/4}}\\
 &\ge1-\frac9{(\log N)^{5/4}}.
\end{align*}
Thus the squared standardized threshold is at least
\(2|S|\log N[1-9/(\log N)^{5/4}]\).  Mills' bound gives the following
intermediate upper bound for the failure probability of this one fixed set:
\[
 \frac{1}{\sqrt{\pi(1-9/(\log N)^{5/4})|S|\log N}}
 N^{-|S|(1-9/(\log N)^{5/4})}.
\]
For all sufficiently large \(N\), this is at most
\begin{equation}\label{eq:inner-one-set}
 \frac{2}{\sqrt{|S|\log N}}
 N^{-|S|(1-9/(\log N)^{5/4})}.
\end{equation}

We now require the estimate for every error set simultaneously.  Fix a
cardinality \(k\).  There are \(\binom Nk\) sets of this size, so the union
bound at this one cardinality is at most
\begin{align*}
 &\binom Nk\,
 \frac{2}{\sqrt{k\log N}}
 N^{-k(1-9/(\log N)^{5/4})}\\
 &\quad\le
 \frac{N^k}{k!}\,
 \frac{2}{\sqrt{k\log N}}
 N^{-k}N^{9k/(\log N)^{5/4}}\\
 &\quad=
 \frac{2}{k!\sqrt{k\log N}}
 \exp\left(\frac{9k}{(\log N)^{1/4}}\right).
\end{align*}
The first line is ``number of sets'' times ``failure probability for one
fixed set.''  The second line uses \(\binom Nk\le N^k/k!\) and separates the
Gaussian power into its leading \(N^{-k}\) term and the relative-slack term.
The factors \(N^k\) and \(N^{-k}\) cancel exactly.  This is the only place
where the constant two is used sharply.

Because \(9/(\log N)^{1/4}\to0\), for all sufficiently large \(N\),
\[
 \exp\left(\frac{9k}{(\log N)^{1/4}}\right)\le2^k.
\]
Therefore the contribution from all sets of cardinality \(k\) is at most
\[
 \frac{2}{\sqrt{\log N}}\frac{2^k}{k!\sqrt{k}}.
\]
Summing this expression over the actual inner cardinalities can only be
smaller than summing it over every positive integer.  Hence the complete
conditional union bound is at most
\[
 \frac{2}{\sqrt{\log N}}
 \sum_{k=1}^{\infty}\frac{2^k}{k!\sqrt{k}}
 =O\left(\frac1{\sqrt{\log N}}\right).
\]
The infinite series is a fixed finite constant, and the factor
\(1/\sqrt{\log N}\) tends to zero.  Finally add the probability that the
sharp channel-sum event fails.  Thus an explicit upper bound for this lemma's
failure probability is
\[
 \Pp\{\text{some sharp channel-sum bound fails}\}
 +\frac{2}{\sqrt{\log N}}
 \sum_{k=1}^{\infty}\frac{2^k}{k!\sqrt{k}},
\]
and both terms tend to zero.
\end{proof}

\paragraph{Where \(2\log N\) is used.}
For an error set of size \(k\), the Gaussian tail supplies essentially
\(N^{-k}\), while the number of sets supplies essentially \(N^k/k!\).
The powers of \(N\) cancel.  The Gaussian prefactor
\(1/\sqrt{k\log N}\) and the factorial \(1/k!\) make the remaining sum
vanish.  This is the threshold-critical calculation.

\subsection{Convert the uniform bounds into objective progress}

We have now controlled the only random quantities appearing in the exact
objective-gap and total-gain identities.  Substituting those bounds into the
identities yields the two deterministic landscape conclusions needed by the
descent argument: positive excess cost and a quantified improving flip.

\begin{lemma}[Excess cost and one-step progress in the inner region]
\label{lem:inner-gain}
Suppose the following two statements hold simultaneously for every inner
error set: the squared channel-sum norm lies between
\([1-(\log N)^{-5/4}]N|S|\) and
\([1+(\log N)^{-5/4}]N|S|\), and the magnitude of the unscaled noise
projection is at most
\([1-4(\log N)^{-5/4}]\sqrt{\rho N}|S|\).  The preceding two lemmas show
that this joint event has probability tending to one.  On this event, for
every error set satisfying
\[
 1\le |S|\le\left\lceil\frac{N}{(\log N)^3}\right\rceil,
\]
\begin{equation}\label{eq:inner-gap}
 \frac{12\rho}{(\log N)^{5/4}}|S|
 \le f(\boldsymbol x)-f(\one)
 \le8\rho|S|,
\end{equation}
and
\begin{equation}\label{eq:inner-total-gain}
 \sum_{i\in S}\bigl[f(\boldsymbol x)-f(\boldsymbol x^{(i)})\bigr]
 \ge\frac{4\rho}{(\log N)^{5/4}}|S|.
\end{equation}
Hence correcting at least one currently wrong bit decreases the cost by at
least \(4\rho/(\log N)^{5/4}\).
\end{lemma}

\begin{proof}
Recall that the objective gap equals
\[
 \frac{4\rho}{N}
 \left\|\sum_{i\in S}\boldsymbol h_i\right\|_2^2
 +4\sqrt{\frac\rho N}\,
 \boldsymbol w^{\mathsf T}\sum_{i\in S}\boldsymbol h_i.
\]
The squared channel-sum norm is at least
\([1-(\log N)^{-5/4}]N|S|\).  The scaled noise projection can be as negative
as
\[
 -\left(1-\frac4{(\log N)^{5/4}}\right)\rho|S|.
\]
Substituting these two worst-case values, the coefficient of \(\rho|S|\) is
\[
 4\left[
 1-\frac1{(\log N)^{5/4}}
 -\left(1-\frac4{(\log N)^{5/4}}\right)
 \right]
 =\frac{12}{(\log N)^{5/4}}.
\]
Using the upper channel bound and the most positive allowed noise projection
gives a coefficient no larger than \(8\).  Hence every nontruth vector in the
inner region lies strictly above the truth, and its excess cost is at
most \(8\rho|S|\).

For one-step progress, the exact total wrong-bit gain is
\[
4\left[
 \frac\rho N\left(
 2\left\|\sum_{i\in S}\boldsymbol h_i\right\|_2^2
 -\sum_{i\in S}\|\boldsymbol h_i\|_2^2\right)
 +\sqrt{\frac\rho N}\,
 \boldsymbol w^{\mathsf T}\sum_{i\in S}\boldsymbol h_i
\right].
\]
Applying the sharp channel statement to singleton sets gives
\[
 \sum_{i\in S}\|\boldsymbol h_i\|_2^2
 \le\left(1+\frac1{(\log N)^{5/4}}\right)N|S|.
\]
To lower-bound the total gain, combine this upper bound with the lower bound
on the squared channel-sum norm and the most negative permitted noise
projection.  Substitution into \eqref{eq:gain-identity} gives
\begin{align*}
 \sum_{i\in S}\Delta_i(\boldsymbol x)
 &\ge4\rho|S|\left[
 2\left(1-\frac1{(\log N)^{5/4}}\right)
 -\left(1+\frac1{(\log N)^{5/4}}\right)
 -\left(1-\frac4{(\log N)^{5/4}}\right)
 \right]\\
 &=\frac{4\rho|S|}{(\log N)^{5/4}}.
\end{align*}
Thus the sum of the \(|S|\) wrong-bit gains is at least
\[
 \frac{4\rho|S|}{(\log N)^{5/4}}.
\]
At least one wrong bit therefore offers the average gain
\[
 \frac{4\rho}{(\log N)^{5/4}}.
\]
At the boundary SNR \(\rho=2\log N\), this is
\(8/(\log N)^{1/4}\).  It may tend slowly to zero, but it remains much larger
than the algorithm's acceptance floor \(1/N\).  No increase of the SNR is
required.
\end{proof}

\begin{keybox}
\textbf{The two-region landscape is complete.}
For every error cardinality from one through
\(\lceil4N/(\log N)^{1/4}\rceil\), exactly one of the two region statements applies.
Every such nontruth vector lies strictly above the truth.  Some wrong bit
offers gain at least \(\rho/2\) in the outer region and at least
\(4\rho/(\log N)^{5/4}\) in the inner region.

The second conclusion is stronger than merely saying that the truth has
lower cost.  It says that every controlled nontruth vector has an adjacent
binary vector, obtained by correcting one currently wrong coordinate, with
strictly smaller cost.  Hence the controlled ball contains no spurious
one-bit local minimum.  The landscape may remain globally nonconvex, and the
greedy path may temporarily flip a correct coordinate, but it cannot become
trapped inside this basin.

These are uniform statements: one event asserts the inequalities for every
eligible set \(S\) simultaneously.  After that event has occurred, any set
visited by the data-dependent greedy path is already covered.  We never
condition on which path the algorithm chooses, and independence between the
path and the landscape is neither asserted nor needed.
\end{keybox}

\section{The deterministic descent argument}\label{sec:descent}

The probabilistic work is now complete.  Fix a realization on which the
warm-start statement and all four uniform geometric statements hold.  From
this point onward, the argument is deterministic.  Written out, the
assumptions are:
\begin{enumerate}[label=\textbf{\arabic*.}]
\item For every outer error set, the squared channel-sum norm lies between
\((3/4)N|S|\) and \((5/4)N|S|\).
\item For every outer error set, the scaled noise projection has magnitude
at most \(\rho|S|/8\).
\item For every inner error set, the squared channel-sum norm lies within the
relative factor \(1\pm(\log N)^{-5/4}\) of \(N|S|\).
\item For every inner error set, the unscaled noise projection has magnitude
at most
\[
 \left[1-\frac4{(\log N)^{5/4}}\right]\sqrt{\rho N}|S|.
\]
\end{enumerate}
The probability sections prove that these four statements and the LMMSE warm
start hold together with probability tending to one.  We now freeze any
realization on which they all hold and show that the actual steepest-descent
path reaches the truth before either stopping rule can activate.

The proof uses the cost \(f\) as a Lyapunov function.  Hamming distance may
move in either direction, but every accepted flip strictly decreases
\(f\).  The uniform geometric events give three deterministic facts at every
point in the controlled ball: every nontruth point has cost above the truth;
the outer boundary has cost higher than the initialization; and some bit
always offers a quantified decrease.  Confinement and termination follow
from these three comparisons.

The \emph{cost budget} at a point \(\boldsymbol x\) is simply its excess cost
above the truth,
\[
 f(\boldsymbol x)-f(\one).
\]
Because every controlled nontruth point has cost above \(f(\one)\), a
downhill path starting from \(\boldsymbol x\) can decrease its cost by at
most this amount before reaching the truth's cost level.  This is why an
upper bound on excess cost, divided by a lower bound on the decrease per
step, bounds the number of remaining steps.

\begin{proposition}[Confinement and termination]\label{prop:descent}
For all sufficiently large \(N\) and \(\rho\ge2\log N\), if
\[
 \dH(\boldsymbol x^{\mathrm{init}},\one)
 \le\left\lceil\frac{N}{(\log N)^{1/4}}\right\rceil
\]
and the conclusions of
\Cref{lem:inner-channel,lem:inner-noise,lem:outer-channel,lem:outer-noise}
hold, then the detector reaches \(\one\) before accepting
\(M_N=4NL_N\) flips, with \(L_N\) as defined in Stage~2.
\end{proposition}

\begin{proof}
If \(\boldsymbol x^{\mathrm{init}}=\one\), the detector is already at the
truth.  The inner-region excess-cost lower bound, applied to every singleton
error set, says that every one-bit neighbor has strictly larger cost.  Hence
every flip from the truth would increase the cost, and the detector stops.  We
henceforth assume that the initialization is not the truth.

\textbf{Step 1: bound the starting excess cost.}
It does not matter for correctness whether the initialization begins in the
inner or outer region; both are contained in the controlled ball.  The
distinction mattered only in the preceding probability proof, where the two
regions required different estimates.  The deterministic argument uses their
common conclusions.  In the inner region we proved
\(f(\boldsymbol x)-f(\one)\le8\rho|S|\).  In the outer region we proved the
smaller upper bound
\(f(\boldsymbol x)-f(\one)\le(11/2)\rho|S|\).  Since the warm start has at
most \(\lceil N/(\log N)^{1/4}\rceil\) errors, either case gives
\begin{equation}\label{eq:start-height}
 f(\boldsymbol x^{\mathrm{init}})-f(\one)
 \le8\rho\left\lceil\frac{N}{(\log N)^{1/4}}\right\rceil.
\end{equation}

\textbf{Step 2: the decreasing-cost path cannot leave the controlled ball.}
Suppose, for the sake of the argument, that the path eventually exited the
controlled ball.  A bit flip changes Hamming distance by exactly one, so
before exiting it would have to visit a vector at distance
\(\lceil4N/(\log N)^{1/4}\rceil\).  The outer-region lower bound says that
the excess cost of every vector on this boundary is at least
\begin{equation}\label{eq:outer-wall}
 f(\boldsymbol x)-f(\one)
 \ge\frac52\rho
 \left\lceil\frac{4N}{(\log N)^{1/4}}\right\rceil.
\end{equation}
This is a lower bound on \(f(\boldsymbol x)-f(\one)\), not on Hamming
distance.  To retain the ceiling terms explicitly, put
\(a_N=N/(\log N)^{1/4}\).  The boundary height is at least
\[
 \frac52\rho\lceil4a_N\rceil\ge10\rho a_N,
\]
whereas \eqref{eq:start-height} is at most
\[
 8\rho\lceil a_N\rceil\le8\rho(a_N+1).
\]
Because \(a_N>4\) for all sufficiently large \(N\), the boundary height is
strictly larger than the starting height.

Every accepted flip lowers \(f\).  Therefore every later iterate has cost at
most the initialization's cost.  But every point through which an escaping
single-bit path would have to pass has cost strictly larger than the
initialization's cost.  Such a point is unreachable by a decreasing-cost
path.  The proposed exit is impossible, and every nontruth iterate remains
in one of the two controlled regions.

\textbf{Step 3: a nontruth vector cannot be a stopping point.}
Every nontruth iterate is now known to lie in one of the two regions.  If it
is in the outer region, some wrong bit offers gain at least \(\rho/2\).  If
it is in the inner region, some wrong bit offers gain at least
\(4\rho/(\log N)^{5/4}\).  The algorithm chooses the largest gain over all
bits, so its selected cost decrease is at least the applicable value.  Since
\(\rho\ge2\log N\),
\[
 \frac{4\rho}{(\log N)^{5/4}}
 \ge\frac8{(\log N)^{1/4}}>\frac1N
\]
for all sufficiently large \(N\).
The detector's ordinary stopping rule is: stop if no bit offers a cost
decrease of at least \(1/N\).  At every nontruth point, the displayed lower
bound provides such a bit.  Therefore the \(1/N\) rule cannot stop the
algorithm at a nontruth vector.  Equivalently, on this good landscape there
is no one-bit local minimum other than the truth anywhere in the controlled
ball.

\textbf{Step 4: recovery occurs within the allowed number of flips.}
At every confined nontruth point, the outer-region gain is at least
\(\rho/2\), while the inner-region gain is at least
\(4\rho/(\log N)^{5/4}\). For all sufficiently large \(N\),
\[
 \frac\rho2\ge\frac{4\rho}{(\log N)^{5/4}}.
\]
Thus, regardless of which region contains the current iterate, every accepted
flip whose starting point is not \(\one\) decreases the objective by at least
\[
 g_N:=\frac{4\rho}{(\log N)^{5/4}}.
\]
This is a local abbreviation for the proof, not a quantity computed by the
detector.

Suppose that \(T\) accepted flips occur without the path reaching \(\one\).
Step~2 confines the endpoint \(\boldsymbol x_T\) to the controlled ball.  Since
it is still a nontruth point, the uniform landscape lower bound gives
\(f(\boldsymbol x_T)>f(\one)\). Summing the decrease over the \(T\) flips and
using \eqref{eq:start-height} therefore yields
\[
 \begin{aligned}
 0<f(\boldsymbol x_T)-f(\one)
 &\le f(\boldsymbol x^{\mathrm{init}})-f(\one)-Tg_N\\
 &\le 8\rho\left\lceil\frac{N}{(\log N)^{1/4}}\right\rceil
 -T\frac{4\rho}{(\log N)^{5/4}}.
 \end{aligned}
\]
Consequently,
\[
 \begin{aligned}
 T
 &<2(\log N)^{5/4}
 \left\lceil\frac{N}{(\log N)^{1/4}}\right\rceil\\
 &\le2N\log N+2(\log N)^{5/4}
 <4N\log N
 \end{aligned}
\]
for all sufficiently large \(N\).

By definition,
\[
 L_N\ge\log_2N=\frac{\log N}{\log2}>\log N,
 \qquad
 M_N=4NL_N>4N\log N.
\]

By Step~3, the algorithm continues to accept flips at every nontruth
iterate.  If it failed to reach \(\one\) by the cap, then \(T=M_N\)
accepted flips would occur without a hit, contradicting
\(T<4N\log N<M_N\).  Hence a first accepted flip \(R\ge1\) whose post-flip
iterate equals \(\one\) exists.  Applying the preceding no-hit inequality
to the first \(R-1\) flips gives
\[
 R<2(\log N)^{5/4}
 \left\lceil\frac{N}{(\log N)^{1/4}}\right\rceil+1
 <4N\log N<M_N
\]
for all sufficiently large \(N\).  Thus the detector cannot exhaust its
allowed number of flips while the iterate is a nontruth vector.

Thus the detector reaches \(\one\) strictly before the flip limit.  At
\(\one\), the singleton
inner-region lower bound says every one-bit neighbor has larger cost, so every
flip has negative gain and the detector stops.
\end{proof}

\section{Proof of polynomial-time recovery and runtime}

\begin{proof}[Proof of \Cref{thm:main}\textup{(a)}]
By \Cref{prop:warm}, the rounded LMMSE initialization satisfies
\[
 \dH(\boldsymbol x^{\mathrm{init}},\boldsymbol x^\star)
 \le\left\lceil\frac{N}{(\log N)^{1/4}}\right\rceil
\]
with probability tending to one.  The outer channel and noise conclusions
follow from \Cref{lem:outer-channel,lem:outer-noise}; the corresponding inner
conclusions follow from \Cref{lem:inner-channel,lem:inner-noise}.  All these conclusions
hold together with probability tending to one by the union bound; the events
need not be independent.

On that intersection, \Cref{prop:descent} is deterministic and gives exact
recovery.  Column-sign symmetry transfers the conclusion from
\(\boldsymbol x^\star=\one\) to every deterministic transmitted vector, with
the same failure probability.  Uniformity over \(\rho\ge2\log N\) follows
from the boundary reductions already made in the proof: the LMMSE matrix
decreases in the positive-semidefinite order as \(\rho\) increases, and the
outer and inner noise events at \(2\log N\) imply their counterparts at every
larger SNR.  Thus the sum of the displayed failure bounds controls the
supremum in \eqref{eq:uniform-risk} and tends to zero.

For the runtime, form
\(\boldsymbol G=\boldsymbol H^{\mathsf T}\boldsymbol H\) and solve the dense ridge system in
\eqref{eq:detector-lmmse}; these operations cost \(O(N^3)\).  Put
\(a=\sqrt{\rho/N}\), maintain the residual
\(\boldsymbol r=\boldsymbol y-a\boldsymbol H\boldsymbol x\), and set
\(\boldsymbol q=\boldsymbol H^{\mathsf T}\boldsymbol r\).  Direct expansion
gives every gain as
\[
 \Delta_j(\boldsymbol x)
 =-4a x_jq_j-4a^2G_{jj}.
\]
After flipping coordinate \(i\), if \(s=x_i\) denotes its pre-flip sign,
\[
 \boldsymbol q'
 =\boldsymbol q+2as\boldsymbol G_{:i}.
\]
Thus updating \(\boldsymbol q\) requires one Gram-column addition.  All
gains are then recomputed and scanned from \(\boldsymbol q\), the current
signs, and the stored diagonal of \(\boldsymbol G\), also in \(O(N)\)
operations.  Hence each accepted flip costs \(O(N)\).  The deterministic
flip limit \(M_N=4NL_N=O(N\log N)\)
therefore makes the descent stage cost \(O(N^2\log N)\).  The dense LMMSE solve still dominates, so the overall
runtime is \(O(N^3)\).
These are unit-cost exact-real arithmetic operations in the model stated in
\Cref{sec:model-result}; no finite-precision or bit-complexity conclusion is
asserted.
\end{proof}

\begin{keybox}
\textbf{Achievability proof complete.}
On the common high-probability event, rounded LMMSE starts inside the
controlled ball, the objective barrier prevents escape, every nontruth
iterate offers a decrease larger than the stopping threshold, and the total
objective budget permits fewer than \(M_N\) accepted flips before recovery.
Thus the detector returns \(\boldsymbol x^\star\) with uniformly vanishing
failure probability and worst-case arithmetic cost \(O(N^3)\).
\end{keybox}

\section{Maximum likelihood fails below the first-order boundary}

The recovery proof is complete.  We now establish the converse in
\Cref{thm:main}(b).  The obstruction is deliberately local: if even one
Hamming-distance-one neighbor has smaller objective value than the
transmitted vector, then the transmitted vector cannot be an ML minimizer.
The same \(2\log N\) nearest-neighbor scale appears in the earlier ML analyses
cited in the introduction \citep{HansenEtAl2009,HassibiEtAl2014}.  The task
here is the converse: show that, with high probability, at least one of these
dependent events occurs.
The proof proceeds through four conceptual reductions, implemented in six
lemmas:
\[
 \begin{gathered}
 \text{one-bit objective gap}
 \ \Longrightarrow\ 
 \text{one short Gaussian interval}
 \ \Longrightarrow\ 
 \text{one bad column}\\[1mm]
 \Longrightarrow\ 
 \text{at least one of the \(N\) columns is bad}.
 \end{gathered}
\]

As before, column-sign symmetry reduces the proof to
\(\boldsymbol x^\star=\one\).  Let \(\boldsymbol h_i\) denote column \(i\)
of \(\boldsymbol H\).

\subsection{Reduce ML failure to a one-bit inequality}

It is enough to identify a one-bit neighbor whose objective is below the
truth.  The first lemma writes that comparison as the sign of a single
scalar.

\begin{lemma}[Exact cost of one wrong flip]\label{lem:converse-one-bit}
Let \(\one^{(i)}\) be obtained from \(\one\) by flipping coordinate \(i\).
Then
\begin{equation}\label{eq:converse-one-bit}
 f(\one^{(i)})-f(\one)=4T_i(\rho),
 \qquad
 T_i(\rho)=\frac{\rho}{N}\|\boldsymbol h_i\|_2^2
 +\sqrt{\frac\rho N}\,\boldsymbol h_i^{\mathsf T}\boldsymbol w.
\end{equation}
Thus \(T_i(\rho)<0\) certifies that the truth is not an ML minimizer.
\end{lemma}

\begin{proof}
At the truth, the residual is \(\boldsymbol w\).  After flipping bit \(i\),
it is \(\boldsymbol w+2\sqrt{\rho/N}\,\boldsymbol h_i\).  Expanding the
difference of the two squared norms gives
\[
 \|\boldsymbol w+2\sqrt{\rho/N}\,\boldsymbol h_i\|_2^2
 -\|\boldsymbol w\|_2^2
 =4\frac\rho N\|\boldsymbol h_i\|_2^2
 +4\sqrt{\frac\rho N}\,\boldsymbol h_i^{\mathsf T}\boldsymbol w.
\]
\end{proof}

It is enough to work at the largest SNR in the converse range.  The next
one-line monotonicity statement explains why.

\begin{lemma}[A bad bit remains bad at smaller SNR]\label{lem:converse-monotone}
For fixed \((\boldsymbol H,\boldsymbol w)\), if
\(T_i(\rho_0)<0\), then \(T_i(\rho)<0\) for every
\(0<\rho\le\rho_0\).
\end{lemma}

\begin{proof}
Divide the displayed expression for \(T_i(\rho)\) in
\eqref{eq:converse-one-bit} by \(\sqrt\rho>0\):
\[
 \frac{T_i(\rho)}{\sqrt\rho}
 =\sqrt\rho\,\frac{\|\boldsymbol h_i\|_2^2}{N}
 +\frac{\boldsymbol h_i^{\mathsf T}\boldsymbol w}{\sqrt N}.
\]
The right side is nondecreasing in \(\rho\).
\end{proof}

For the rest of the converse, set
\begin{equation}\label{eq:converse-boundary-rho}
 \rho_0=2\log N-\log\log N-s_N.
\end{equation}
Because \(s_N=o(\log N)\), this value is positive for all sufficiently
large \(N\), and \(\rho_0\sim2\log N\).

\subsection{Construct a lower-cost neighbor from a Gaussian interval}

Condition temporarily on a fixed noise vector.  For a standard Gaussian
column \(\boldsymbol h\) independent of \(\boldsymbol w\), its signed projection onto the noise
direction is
\begin{equation}\label{eq:converse-projection}
 g=\boldsymbol h^{\mathsf T}
 \frac{\boldsymbol w}{\|\boldsymbol w\|_2}.
\end{equation}
Conditional on \(\boldsymbol w\), rotational invariance gives
\(g\sim\mathcal N(0,1)\).  We use the relative-error scale
\(N^{-1/4}\): it tends to zero, its chi-square exponent
\(N(N^{-1/4})^2=\sqrt N\) diverges, and
\(\rho_0N^{-1/4}\to0\).  These are precisely the three properties needed
below.

\begin{lemma}[A sufficient Gaussian interval]\label{lem:converse-slab}
Suppose
\begin{equation}\label{eq:converse-typical-lengths}
 \|\boldsymbol w\|_2^2\ge(1-N^{-1/4})N,
 \qquad
 \|\boldsymbol h\|_2^2\le(1+N^{-1/4})N.
\end{equation}
If
\begin{equation}\label{eq:converse-slab}
 -(1+5N^{-1/4})\sqrt{\rho_0}-\frac1{\sqrt{\rho_0}}
 \le g\le
 -(1+5N^{-1/4})\sqrt{\rho_0},
\end{equation}
then the corresponding \(T_i(\rho_0)\) is negative for all sufficiently
large \(N\).
\end{lemma}

\begin{proof}
Using
\(\boldsymbol h^{\mathsf T}\boldsymbol w=g\|\boldsymbol w\|_2\),
the inequality \(T_i(\rho_0)<0\) is exactly
\begin{equation}\label{eq:converse-threshold}
 g<-\sqrt{\frac{\rho_0}{N}}
 \frac{\|\boldsymbol h\|_2^2}{\|\boldsymbol w\|_2}.
\end{equation}
Under \eqref{eq:converse-typical-lengths}, the magnitude of the threshold on
the right is at most
\[
 \sqrt{\rho_0}\,
 \frac{1+N^{-1/4}}{\sqrt{1-N^{-1/4}}}
 \le(1+4N^{-1/4})\sqrt{\rho_0}
\]
for all sufficiently large \(N\).  Every \(g\) in
\eqref{eq:converse-slab} is no greater than
\(-(1+5N^{-1/4})\sqrt{\rho_0}\), which is strictly below the required
threshold.
\end{proof}

The interval has width \(1/\sqrt{\rho_0}\).  Keeping that width is what
produces the \(-\log\log N\) correction.

All asymptotic remainders in the next three probability bounds depend only
on the deterministic sequences \(N\) and \(\rho_0\), and hence on the fixed
sequence \(s_N\), but not on the conditioned noise vector.  Taking a
decreasing envelope of the absolute values of the three deterministic
\(o(1)\) remainders, fix a deterministic sequence
\(\varepsilon_N\downarrow0\) large enough that all three bounds below hold;
in particular, \(e^{-1}-\varepsilon_N>0\) eventually.

\begin{lemma}[Probability of the sufficient interval]
\label{lem:converse-slab-probability}
For a standard Gaussian \(g\),
\begin{equation}\label{eq:converse-slab-probability}
 \Pp\{g\text{ satisfies \eqref{eq:converse-slab}}\}
 \ge(e^{-1}-\varepsilon_N)
       \frac{\phi(\sqrt{\rho_0})}{\sqrt{\rho_0}}.
\end{equation}
\end{lemma}

\begin{proof}
On the negative half-line, the smallest density in the interval occurs at
its left endpoint.  Hence interval length times minimum density gives
\[
 \Pp\{g\text{ satisfies \eqref{eq:converse-slab}}\}
 \ge\frac1{\sqrt{\rho_0}}
 \phi\!\left((1+5N^{-1/4})\sqrt{\rho_0}
             +\frac1{\sqrt{\rho_0}}\right).
\]
The square of the displayed argument minus \(\rho_0\) equals
\[
 10N^{-1/4}\rho_0+25N^{-1/2}\rho_0
 +2+10N^{-1/4}+\frac1{\rho_0}=2+o(1).
\]
Since \(\phi(t)=(2\pi)^{-1/2}e^{-t^2/2}\), the density ratio is
\(e^{-1+o(1)}\).  By the deterministic choice of \(\varepsilon_N\), it is
at least \(e^{-1}-\varepsilon_N\) for all sufficiently large \(N\), proving
the claim.
\end{proof}

The projection \(g\) and the column length are not independent.  The next
lemma avoids any independence claim by subtracting the probability of an
atypically long column.

\begin{lemma}[One-column probability]\label{lem:converse-one-column}
For every fixed \(\boldsymbol w\) satisfying the first inequality in
\eqref{eq:converse-typical-lengths},
\begin{equation}\label{eq:converse-one-column-probability}
 \Pp\{T_i(\rho_0)<0\mid\boldsymbol w\}
 \ge(e^{-1}-\varepsilon_N)
       \frac{\phi(\sqrt{\rho_0})}{\sqrt{\rho_0}},
\end{equation}
uniformly over all such \(\boldsymbol w\).
Here only the independent column \(\boldsymbol h_i\) remains random under
the conditional probability.
\end{lemma}

\begin{proof}
Let \(A\) be the event that the column-length inequality in
\eqref{eq:converse-typical-lengths} holds, and let \(C\) be the Gaussian
interval event \eqref{eq:converse-slab}.  The preceding sufficient-interval
lemma gives \(A\cap C\subseteq\{T_i(\rho_0)<0\}\).  For arbitrary events,
\(\Pp(A\cap C)\ge\Pp(C)-\Pp(A^c)\).  The chi-square bound
\eqref{eq:chi-tail} therefore gives
\begin{align*}
 \Pp\{T_i(\rho_0)<0\mid\boldsymbol w\}
 &\ge \Pp(C)-2e^{-\sqrt N/8}.
\end{align*}
The preceding density calculation shows, deterministically, that
\[
 \frac{\Pp(C)}{\phi(\sqrt{\rho_0})/\sqrt{\rho_0}}
 \longrightarrow e^{-1}.
\]
Because \(\rho_0<2\log N\),
\[
 \frac{\phi(\sqrt{\rho_0})}{\sqrt{\rho_0}}
 \ge\frac1{N\sqrt{4\pi\log N}}.
\]
The stretched-exponentially small term is therefore negligible relative to
\(\phi(\sqrt{\rho_0})/\sqrt{\rho_0}\).  Both deterministic convergence
errors are absorbed in the chosen \(\varepsilon_N\), which proves
\eqref{eq:converse-one-column-probability} uniformly in \(\boldsymbol w\).
\end{proof}

\paragraph{Where the \(\log\log N\) correction comes from.}
One column is bad with probability on the scale
\[
 \frac{\phi(\sqrt\rho)}{\sqrt\rho}
 =\frac{e^{-\rho/2}}{\sqrt{2\pi\rho}}.
\]
The factor \(e^{-\rho/2}\) produces \(2\log N\).  The interval has width
\(1/\sqrt\rho\); after multiplication by \(N\), this factor produces the
additional \(-\log\log N\) term.  Dropping the prefactor would lose the
lower-order boundary proved here.

\subsection{Use conditional independence across the columns}

Conditional on the common noise vector, different one-bit tests use
different independent channel columns.  This turns the one-column estimate
into \(N\) independent opportunities for a lower-cost neighbor.  It remains
to bound the probability that none of them occurs and then evaluate the
result at \(\rho_0\).

\begin{lemma}[At least one lower-cost neighbor appears]\label{lem:converse-amplification}
For every \(\boldsymbol w\) satisfying the first inequality in
\eqref{eq:converse-typical-lengths},
\begin{equation}\label{eq:converse-no-bad-column}
 \Pp\{T_i(\rho_0)\ge0\text{ for every }i\mid\boldsymbol w\}
 \le
 \exp\!\left\{-(e^{-1}-\varepsilon_N)
 N\frac{\phi(\sqrt{\rho_0})}{\sqrt{\rho_0}}\right\}.
\end{equation}
\end{lemma}

\begin{proof}
Given \(\boldsymbol w\), the \(N\) events depend on independent and
identically distributed columns.  If the probability that one column gives
\(T_i(\rho_0)<0\) is \(p\), the probability that no column does so is
\((1-p)^N\le e^{-Np}\).
Apply \Cref{lem:converse-one-column}.
\end{proof}

\begin{proof}[Proof of \Cref{thm:main}\textup{(b)}]
First take \(\boldsymbol x^\star=\one\) and \(\rho=\rho_0\).  Split according
to whether \(\|\boldsymbol w\|_2^2\ge(1-N^{-1/4})N\).  The chi-square bound
and \Cref{lem:converse-amplification} give
\begin{align}
 &\Pp\{T_i(\rho_0)\ge0\text{ for every }i\}\notag\\
 &\quad\le2e^{-\sqrt N/8}
 +\exp\!\left\{-(e^{-1}-\varepsilon_N)
 N\frac{\phi(\sqrt{\rho_0})}{\sqrt{\rho_0}}\right\}.
 \label{eq:converse-final-bound}
\end{align}
It remains to evaluate the exponent.  Substituting
\eqref{eq:converse-boundary-rho} gives
\begin{align*}
 N\frac{\phi(\sqrt{\rho_0})}{\sqrt{\rho_0}}
 &=\frac{N}{\sqrt{2\pi\rho_0}}e^{-\rho_0/2}\\
 &=\frac{1+o(1)}{\sqrt{4\pi}}e^{s_N/2}
 \longrightarrow\infty.
\end{align*}
Both terms in \eqref{eq:converse-final-bound} vanish.  With probability
tending to one, some \(T_i(\rho_0)<0\), and
\Cref{lem:converse-one-bit} then gives a strictly lower-cost one-bit
neighbor.

In particular,
\[
 \{\one\in\arg\min_{\boldsymbol x}f(\boldsymbol x)\}
 \subseteq\bigcap_{i=1}^N\{T_i(\rho_0)\ge0\},
\]
so the preceding probability bound is exactly sufficient for the claimed
ML-minimizer failure.

By \Cref{lem:converse-monotone}, the same coupled event implies failure at
every \(0<\rho\le\rho_0\).  Column-sign symmetry gives the same probability
for every deterministic transmitted word.  Consequently the bound controls
the supremum in \eqref{eq:uniform-converse-risk}, proving the claimed uniform
ML failure.
\end{proof}

\begin{keybox}
\textbf{Converse proof complete.}
At or above \(2\log N\), the explicit two-stage detector achieves exact block
recovery with uniformly vanishing error probability.  At or below
\(2\log N-\log\log N-s_N\), even one-bit competitors prevent ML recovery
with uniformly high probability.  The two statements meet at the sharp
first-order constant \(2\); the remaining lower-order window is not resolved
here.
\end{keybox}

\appendix
\section{Decision-theoretic consequence of the converse}
\label{app:decision}

First-order ML achievability was already identified in
\citet{HansenEtAl2009} and stated with a diverging additive slack in
\citet[Lemma~IV.2]{HassibiEtAl2014}.  The new constructive content of this
paper is therefore \Cref{thm:main}\textup{(a)}, not another proof of the known
first-order ML upper bound.  This appendix records the standard
decision-theoretic implication needed only if ``statistical threshold'' is
read in the minimax sense.  It is not used anywhere in the constructive proof.

\begin{corollary}[Minimax consequence]
\label{cor:minimax}
Write \(\Pp_{\boldsymbol x^\star,\rho}\) for probability under
\eqref{eq:model} with the displayed deterministic word and SNR.  For the
sequence in \Cref{thm:main}\textup{(b)}, put
\(\rho_{0,N}=2\log N-\log\log N-s_N\) and define the minimax risk at a
fixed SNR by
\[
 R_N^\star(\rho)
 =\inf_{\delta_N}\ \sup_{\boldsymbol x^\star\in\{\pm1\}^N}
 \Pp_{\boldsymbol x^\star,\rho}
 \{\delta_N(\boldsymbol H,\boldsymbol y,\rho)
       \ne\boldsymbol x^\star\},
\]
where the infimum is over measurable deterministic decision rules and
randomized Markov kernels from
\((\boldsymbol H,\boldsymbol y,\rho)\) to \(\{\pm1\}^N\).  For a randomized
detector, the displayed probability also averages over its independent
internal randomization.  Then
\begin{equation}\label{eq:minimax-lower-risk}
 \lim_{N\to\infty}\ \inf_{0<\rho\le\rho_{0,N}}
 R_N^\star(\rho)=1.
\end{equation}
Thus the statistical threshold and the polynomial arithmetic-complexity
threshold have the same first-order coefficient \(2\).  The unit-cost
exact-real qualification applies to the algorithmic achievability statement,
not to the statistical lower bound.
\end{corollary}

\begin{proof}
Give \(X\in\{\pm1\}^N\) the uniform prior, independently of
\((\boldsymbol H,\boldsymbol w)\).  Conditional on \(\boldsymbol H\),
\[
 p(\boldsymbol y\mid\boldsymbol H,X=\boldsymbol x)
 \propto\exp\{-f(\boldsymbol x)/2\}.
\]
Consequently, least-squares ML is a MAP rule and minimizes the
uniform-prior average block-error probability.  Define
\[
 \eta_N=
 \sup_{\substack{\boldsymbol x\in\{\pm1\}^N\\0<\rho\le\rho_{0,N}}}
 \Pp_{\boldsymbol x,\rho}
 \{\boldsymbol x\in\arg\min_{\boldsymbol z}f(\boldsymbol z)\}.
\]
By \Cref{thm:main}\textup{(b)}, \(\eta_N\to0\).  At any fixed SNR in
the lower regime, the uniform-prior average success probability of MAP is
at most \(\eta_N\).  Bayes optimality implies that the average success
probability of every detector is also at most \(\eta_N\).  Hence, for
every detector,
\[
 \sup_{\boldsymbol x}
 \Pp_{\boldsymbol x,\rho}\{\delta_N\ne\boldsymbol x\}
 \ge 2^{-N}\sum_{\boldsymbol x}
 \Pp_{\boldsymbol x,\rho}\{\delta_N\ne\boldsymbol x\}
 \ge1-\eta_N.
\]
The bound is uniform for \(0<\rho\le\rho_{0,N}\), proving
\eqref{eq:minimax-lower-risk}.
\end{proof}

\begin{remark}[The ML upper endpoint]
The preceding corollary concerns the lower, minimax side only.  For
completeness, \Cref{thm:main}\textup{(a)} also implies
\[
 \lim_{N\to\infty}
 \sup_{\substack{
       \boldsymbol x^\star\in\{\pm1\}^N\\
       \rho\ge2\log N}}
 \Pp_{\boldsymbol x^\star,\rho}
 \{\widehat{\boldsymbol x}_{\rm ML}\ne\boldsymbol x^\star\}=0
\]
for a measurable least-squares ML selector.  This is a risk comparison,
not an eventwise implication from successful local descent.  Indeed, under
the uniform prior used above, MAP optimality makes the average ML error no
larger than the constructive detector's average error.  Objective ties have
probability zero: conditional on an almost surely nonsingular
\(\boldsymbol H\), every pairwise cost difference is a nondegenerate affine
Gaussian function of \(\boldsymbol w\).  Finally, column-sign symmetry makes
the ML error probability identical for every deterministic transmitted word,
so its average and worst-word risks coincide.  This observation is included
only for completeness; first-order ML achievability was already known from
the works cited at the start of this appendix.
\end{remark}

\section{Equivalent formulations and adjacent literatures}
\label{app:adjacent}

The observation law itself has several exact non-MIMO descriptions.
Writing \(A=H/\sqrt N\), it is the load-one uncoded binary-CDMA model
with user energy \(\rho\), and it is the \(B=1\), \(L=n=N\) signed
Gaussian superposition construction
\citep{JosephBarron2012,LiuHsiehVenkataramanan2024}.  Moreover, if
\(b^\star=(x^\star+\boldsymbol 1)/2\), then
\[
 y+\sqrt{\rho/N}\,H\boldsymbol 1
 =H\bigl(2\sqrt{\rho/N}\,b^\star\bigr)+w,
 \qquad b^\star\in\{0,1\}^N,
\]
so the same data also form a dense binary Gaussian-regression problem.
These are identities of the finite-dimensional joint observation law,
not merely analogies between codeword or covariance distributions.

\subsection*{Gaussian random linear estimation and AMP}

The AMP connection is also an exact model identification, once the prior is
specified.  Divide \eqref{eq:model} by \(\sqrt{\rho}\) and write
\[
 \widetilde y=\frac{1}{\sqrt N}H x^\star+\widetilde w,
 \qquad
 \widetilde w\sim\mathcal N(0,\Delta I_N),
 \qquad \Delta=\frac1\rho.
\]
If the coordinates of \(x^\star\) are assigned an i.i.d. Rademacher prior,
this is precisely the scalar \(B=1\), load \(\alpha=1\) Gaussian
random-linear-estimation model studied in the AMP, CDMA, and replica
literatures \citep{Tanaka2002,BayatiMontanari2011,
JavanmardMontanari2013,BarbierMacrisDiaKrzakala2020}.  The theorem in this
paper is instead stated as a worst-word risk bound for every deterministic
\(x^\star\); column-sign symmetry is what makes that formulation compatible
with the usual Rademacher-prior viewpoint.

The earlier results rigorously establish state evolution for fixed iteration
and fixed model parameters.  For scalar discrete priors, the
replica-symmetric formulas for normalized mutual information and MMSE and,
outside the hard phase, AMP optimality in normalized mean-square error hold
under the at-most-three-stationary-points condition on the
replica-symmetric potential stated in
\citet{BarbierMacrisDiaKrzakala2020}; the Rademacher model here falls under
that conclusion when the stated condition holds.  At the good low-noise fixed point, the scalar-channel prediction has hard-decision error
\(Q(\sqrt{\rho})\) to first order.  If one informally treats the coordinate
errors as rare opportunities, then
\[
 NQ(\sqrt{\rho})\asymp1
 \quad\Longleftrightarrow\quad
 \rho=2\log N-\log\log N+O(1),
\]
which explains why AMP and replica calculations expose the same threshold
scale as the one-bit-neighbor calculation in the introduction.  For a recent
discussion of this prediction and its relationship to the present
block-recovery question, see \citet{Krzakala2026MIMOAMP}.

The logical gap is the order of limits and the loss function.  A statement
that first fixes \(\Delta>0\) and the iteration index, lets \(N\to\infty\),
and controls an empirical mean does not automatically remain valid when
\(\Delta_N\asymp1/\log N\), nor does an \(o(1)\) normalized error imply that
all \(N\) coordinates are correct.  One direct AMP route based on a
coordinatewise union bound would require a nonasymptotic or triangular-array
specialization strong enough to give simultaneous \(o(1/N)\)-scale
coordinate tails and to control the required number of iterations; a
different route would need comparably strong all-coordinate control.
Nonasymptotic AMP
representations such as \citet{BaoHanXu2025} make that route plausible, but
we are not aware of a published argument that carries out this specialization
at \(\rho_N=2\log N\).

\subsection*{Comparison with \(\mathbb Z_2\) synchronization}

Gaussian \(\mathbb Z_2\) synchronization is a close planted-recovery
analogue, but it is not the same observation experiment.  In one standard
normalization it observes a symmetric spiked-Wigner matrix
\[
 M=\lambda v^\star v^{\star\mathsf T}+W,
 \qquad
 v_i^\star\in\{\pm N^{-1/2}\},
\]
where \(W=W^{\mathsf T}\) is GOE-scaled: independently for \(i<j\),
\(W_{ij}\sim\mathcal N(0,1/N)\), and independently on the diagonal,
\(W_{ii}\sim\mathcal N(0,2/N)\).
Thus the data are noisy pairwise products of the unknown signs.  By contrast,
the MIMO model observes \(N\) noisy linear measurements through a known
rectangular design.  Even the matched-filter statistic
\(H^{\mathsf T}y\) contains a Wishart Gram matrix and noise correlated
through that same design; it is not a spiked-Wigner observation.

The synchronization literature nevertheless provides important conceptual
precedents: sharp exact-recovery thresholds can be attained by single-stage
polynomial-time relaxations, in community recovery and in synchronization
\citep{HajekWuXu2016,FeiChen2019}, and AMP admits finite-sample analyses in the
spiked-Wigner model \citep{LiFanWei2023}.  Closely related ``rough recovery
followed by local cleanup'' architectures also appear in exact community
recovery \citep{AbbeBandeiraHall2016}.  These works help explain
why the proof strategy here is natural.  They do not directly imply the MIMO
theorem, because a transfer would still have to map the observation law,
the SNR normalization, the whole-vector loss, and the algorithmic trajectory.

The closest coding and CDMA results do not provide the block-recovery
theorem needed here.  \citet{JosephBarron2012} identify the \(B=1\)
signed construction as a Gaussian generator-matrix code, but their
reliability theorem assumes \(B=L^a\) with a positive section-size
exponent.  Thus it does not include \(B=1\); at that endpoint their
least-squares decoder searches all sign vectors, and their outer-code
conversion to block reliability changes the message set.  Likewise,
the nonasymptotic AMP theorem of \citet{RushVenkataramanan2019} controls
a positive section-error fraction for sparse superposition codes at
fixed rate with growing section size; it cannot be specialized by
setting the section alphabet to two to obtain exact recovery here.
The uncoded specialization of \citet{LiuHsiehVenkataramanan2024} has
exactly the present load-one CDMA law, but its state-evolution theorem
characterizes fixed-iteration empirical user- and bit-error rates at
fixed system parameters.  It does not control the event that all users
are correct when the energy grows as \(\rho_N\asymp\log N\).

The regression identity gives another close comparison.  For balanced
words, which are included in our uniform theorem, it has
\(p=n=N\), support size \(s=N/2\), and nonzero coefficient
\(a_N=2\sqrt{\rho/N}\).  The polynomial selector of
\citet{NdaoudTsybakov2020} splits the observations into estimation and
thresholding samples; its estimation-sample condition involves an
unspecified absolute constant, and its guarantee is not shown to cover the
dense endpoint \(s=p/2\) with \(n=p\).
The iterative compressed-sensing theorem of \citet{GaoZhang2022}
avoids that split, but assumes both \(\limsup s/p<1/2\) and
\(s(\log p)^4=o(n)\), excluding \(s=p/2\) and \(n=p=N\).
The polynomial-time correlation guarantee of
\citet{MazumdarSangwan2025} specializes to an \(N\log N\) observation
scale in the dense case, while its sharper linear-regression likelihood
result uses exhaustive fixed-weight decoding.  Finally, the lattice
method of \citet{GamarnikKizildagZadik2021} gives bit-polynomial exact
recovery for dense rational regression, but at \(n=p=N\) its proved
additive-noise tolerance is exponentially small.  In the corresponding
variance-one-design normalization, the noise standard deviation here is
\(\sqrt{N/\rho}=\Theta(\sqrt{N/\log N})\).

There is also an exact conditional negative-Wishart statistic.  Put
\(u=x^\star/\sqrt N\), let
\(U_y\in\mathbb R^{(N-1)\times N}\) have orthonormal rows spanning
\(y^\perp\), and form \(B=U_yH\).  Gaussian conditioning gives,
conditional on \(y\), independent rows
\[
 B_{k,:}\sim
 \mathcal N\!\left(0,
 I_N-\frac{\rho}{1+\rho}uu^{\mathsf T}\right),
 \qquad k=1,\ldots,N-1.
\]
Thus \(B\) is a negative-spiked Rademacher Wishart instance with
\(p=N\), \(n_{\rm samp}=N-1\), and
\(\beta_N=-\rho/(1+\rho)\).  At \(\rho\asymp2\log N\), its residual
variance in the planted direction is
\(1+\beta_N=(1+\rho)^{-1}\asymp(2\log N)^{-1}\).
The statistic contains information about the planted direction only up to a
global sign, which the original observation can orient by comparing the two
residual costs.
This is a one-way statistic, not an equivalence of experiments, because
the projection may discard information.

The neighboring Wishart literature treats different tasks or different
parts of the hard edge.  \citet{BandeiraKuniskyWein2020} study detection
and constrained-PCA certification, including low-degree evidence rather
than a general computational-hardness theorem; they do not give a spike
recovery algorithm for the triangular sequence above.
\citet{ZadikSongWeinBruna2022} use lattice basis reduction to recover a
planted hypercube at the exactly singular endpoint \(\beta=-1\), where
the planted vector is an exact null direction.  Their guarantee relies
on that noiseless algebraic structure.  \citet{Venkat2023} gives
lattice-based discrepancy certificates and negative-Wishart detection
in a much more nearly singular regime; at square aspect its guarantee
requires exponentially small residual variance and does not recover the
spike.  None therefore yields recovery at
\(\beta_N=-1+\Theta(1/\log N)\).

Accordingly, we do not claim novelty for the observation law, its
coding, CDMA, or regression interpretations, the conditional Wishart
statistic, the AMP/replica prediction of the threshold scale, the broad
warm-start-and-cleanup architecture, or local search itself.  To our
knowledge,
no prior published theorem gives a detector that, using polynomially
many unit-cost exact-real operations, has vanishing whole-vector error
for the i.i.d. square signed-Gaussian code uniformly over every
deterministic transmitted word and every \(\rho\ge2\log N\).

\section*{Acknowledgments}

The author thanks Kangwook Lee, Jason D.~Lee, and Samet Oymak for early
conversations and encouragement; Mahdi Soltanolkotabi for discussions of AMP
and fixed-signal versus uniform guarantees; Jy-yong Sohn for reading an early
draft; Christos Thrampoulidis for discussions of box relaxation and related
literature; and Andrea Montanari, Florent Krzakala, and Lenka Zdeborov\'a for
public comments that sharpened the comparison with AMP, rigorous replica
formulas, and \(\mathbb Z_2\) synchronization.  The author thanks David Tse for
historical perspective on the problem and for emphasizing its uncoded scope;
Babak Hassibi for encouraging feedback and for raising questions about
extensions and the MCMC perspective; and Alexandros G.~Dimakis for discussing
this problem with the author during the author's PhD and for his mentorship.

\begingroup
\small
\bibliographystyle{plainnat}
\bibliography{references}
\endgroup

\end{document}